\documentclass{article}
\usepackage{graphicx} 
\usepackage[affil-it]{authblk}
\usepackage[utf8]{inputenc}
\usepackage[english]{babel}
\usepackage{titlesec}
\usepackage[margin=1in]{geometry}
\usepackage{tikz}
\usetikzlibrary{decorations.pathreplacing}
\usepackage{mathtools}
\usepackage{chngcntr}
\usepackage{lineno}
\usepackage{enumerate}
\usepackage{adjustbox}
\usepackage{titlesec}
\usepackage[utf8]{inputenc}
\usepackage{mathrsfs}
\usepackage{cite}
\usepackage{cite}

\usepackage{enumitem}

\usepackage{amsthm}

\newtheorem{theorem}{Theorem}[section]

\usepackage[hyphens]{url}
\usepackage{hyperref}
\usepackage{float}
\usepackage[section]{placeins}
\usepackage[dvipsnames]{xcolor}
\usepackage{soul}
\usepackage[labelfont=bf,font=small]{caption}
\usepackage{bm}
\usepackage[title]{appendix}
\usepackage{graphicx}
\usepackage{tabularx}
\usepackage{subcaption}
\usepackage{float}
\usepackage{mathtools}
\usepackage{amsmath,amssymb,amsthm}
\usepackage{oubraces}
\numberwithin{equation}{section}
\usepackage{chngcntr}
\usepackage{mathrsfs}
\usepackage{multirow}
\usepackage{chngcntr}
\usepackage{parskip}
\usepackage{accents}
\usepackage{soul}
\usepackage{cleveref}
\usepackage{orcidlink}
\usepackage{tikz}

\usetikzlibrary{shapes,arrows,calc}
\usetikzlibrary{positioning}
\tikzstyle{Pblock} = [rectangle, draw, very thick,text width=2cm, text centered, rounded corners, minimum height=1cm]
\tikzstyle{line} = [draw,-latex',very thick]
\tikzstyle{block} = [rectangle, draw, very thick,text width=1cm, text centered, rounded corners, minimum height=1cm]
\tikzstyle{circ} = [draw, very thick, text width=2cm, ellipse,minimum width =2.5cm,minimum height=1.75cm,align=center]
\tikzstyle{svect} = [draw,fill=white,diamond,aspect=2,scale=.8]
\tikzstyle{cvect} = [draw,fill=red!20,diamond,aspect=2,scale=.8]
\tikzstyle{ivect} = [draw,fill=red,diamond,aspect=2,scale=.8]
\tikzstyle{sc} = [draw,fill=white,circle]
\tikzstyle{mc} = [draw,fill=red!10,circle]
\tikzstyle{ac} = [draw,fill=orange!20,circle]
\tikzstyle{rc} = [draw,fill=blue!10,circle]

\newcommand{\model}{\{\eqref{eq:within_host},\eqref{eq:population_level}\}}

\newcommand{\modelsp}{\{\eqref{eq:within_host},\eqref{eq:population_level}\} }

\author{ }
\date{}
\begin{document}

\title{A multi-scale immuno-epidemiological behavioral model for influenza-like illness: connecting scales through symptom scores}
\author{Binod Pant \orcidlink{0000-0001-8619-1861}$^{a}$\footnote{Corresponding author: Email: {\it bpant@math.msstate.edu}}, Summer Atkins$^{b}$, Necibe Tuncer$^{c}$ and Hana M.\ Dobrovolny$^{d}$
\begin{flushleft}
{\it {\small $^{a}$  Department of Mathematics and Statistics, Mississippi State University, Mississippi State, MS 39762, USA}}\newline
{\it {\small $^{b}$  Department of Mathematics, The University of Alabama in Huntsville, Huntsville, AL, 35899, USA}}\newline
{\it {\small $^{c}$  Department of Mathematics and Statistics, Florida Atlantic University, Boca Raton 33431,
Florida, USA}}\newline
{\it {\small $^{d}$  Department of Physics \& Astronomy, Texas Christian University, Fort
Worth, TX, 76129, USA}}\newline
\end{flushleft}
\vspace{-1.5cm}
    }
  \date{}
\maketitle
\noindent
\vspace{1cm}
\begin{abstract}
For many infectious diseases, behavior change is not a population-level reaction to rising case counts, but a personal one, triggered by how sick an individual feels. Yet most models fail to capture how illness severity drives this response. Multi-scale models that link the immunological scale to the epidemiological scale seem well-suited for capturing this behavior. Most multi-scale models link within-host viral load to population-level transmission rates through linking functions that are a function of viral load. However, viral load is often misaligned with symptom severity and thus possibly ill-suited for the purpose of informing behavior at a population level. We present a novel multi-scale model in which symptomatic individuals dynamically modify their transmission based on systemic symptom scores (a composite measure of muscle ache, fatigue, headache, and feverishness informed by within-host immune dynamics) while asymptomatic individuals do not change behavior. We find that stronger illness-driven behavior change delays and lowers the epidemic peak.  
Moreover, the resulting epidemic trajectory remains qualitatively similar to that of a standard SEIR model without behavior change, across several distinct choices of linking functions connecting the two scales. This stands in contrast to models where behavior change is driven by population-level feedback (e.g., reported cases or deaths), which can produce qualitatively distinct signatures such as plateaus, shoulders, or an elongated epidemic decline. In the case when an exponential function of symptom score is used as a linking function, we show that the standard SEIR model can closely reproduce the total number of infectious individuals at a given point in time generated through the multi-scale level. Overall, this study suggests that behavior change driven by symptom severity can be present at the population level without leaving a clear, distinguishing trace in the epidemic curve.
\end{abstract}
{\it Keywords: Immuno-epidemiological model; Multi-scale modeling; Behavioral epidemiology; Symptom score; Linking function.}

\section{Introduction}
Emerging and reemerging infectious diseases have been documented to cause human behavior changes. For example, by examining population-level survey collected human behavior data, Urmi et al. \cite{urmi2024characterizing} showed that during the COVID-19 pandemic, changes in human behavior were correlated with changes in reported mortality and hospitalization. During the COVID-19 pandemic, widespread public mask usage was adopted in the United States for the first time since localized mask mandates were implemented during the 1918 influenza pandemic, though unlike 1918, mask wearing was not enforced through legal ordinances in some cities like San Francisco \cite{cnn_1918_masks_2020}.
Similarly, empirical and modeling studies have attributed the rapid decline of the 2022 mpox outbreak among men who have sex with men (MSM) to behavioral adaptations within affected communities \cite{prochazka2024temporary,maniscalco2024adaptive}.

This empirical evidence largely focuses on population-level observations, and perhaps for this reason, much of the behavioral epidemiology modeling literature has similarly focused on how populations react to population-level aggregate signals such as reported cases, hospitalizations, and deaths
\cite{pant2024behavior,pant2026paradox, kiss2010impact, oveson2025modeling,lejeune2024mathematical}.
However, human behavior change, especially in recurring respiratory diseases (such as flu and RSV), can also occur due to an individual getting sick: symptomatic individuals often alter their activities due to illness severity, unwillingness to infect others, or simple inability to maintain normal routines. In fact, McLean et al. \cite{mclean2017school} report that in the United States, influenza is responsible for 47\% of missed school days attributed to acute respiratory illness.
In contrast, asymptomatic individuals, who unless tested, do not know that they have the disease, tend not to modify their behavior (i.e., contact patterns) and unknowingly transmit diseases. Various modeling studies have shown that the major drivers of the COVID-19 pandemic were asymptomatic individuals (as opposed to symptomatic individuals) \cite{moghadas2020implications,ma2021global, pant2024hetro,dobrovolny20asymp}, largely because those who are symptomatically infectious tend to reduce their contact patterns.

 This illness-induced behavior change, while epidemiologically significant, remains understudied from a modeling perspective compared to {\it population-level models} that track behavior change driven by population-level information. {\it Immuno-epidemiological models} have been developed to connect the within-host and population scales, {but existing models of this kind largely represent the between-host transmission rate (a product of probability of transmission per contact, and the number of contacts per unit time) as a function of viral load \cite{gilchrist2002modeling,martcheva2015introduction}. Since these traditional multi-scale models do not explicitly mention behavior change, it can be deduced that viral load is implicitly being used to account for changes in infectiousness over the course of infection. In this study, we instead focus on how changes in the number of contacts can be informed by information generated by the within-host model. However, for the purpose of understanding illness-induced behavior change,} viral load may not be the most appropriate quantity to inform behavior: viral load dynamics are temporally misaligned with symptom onset \cite{he2020temporal,canini2011population,challenger2022modelling}, and symptom severity often reflects unregulated inflammatory responses rather than direct viral burden \cite{argyropoulos2020association}.



To address these gaps, we develop a novel immuno-epidemiology model that mechanistically links within-host viral dynamics and immune responses to population-level transmission patterns through {\it systemic symptom scores} (a composite measure of muscle ache, fatigue, headache, and feverish feeling) \cite{canini2011population}.
Our approach uses a nested modeling framework, {first introduced in \cite{gilchrist2002modeling}}, 
where a within-host immunology model tracks viral load and pro-inflammatory cytokine levels, among other responses, to generate dynamic symptom scores that directly inform behavioral changes in symptomatic individuals at the population level, while asymptomatic individuals maintain unchanged transmission rates. Unlike previous population-level compartmental approaches that split different stages of symptomatic individuals into different compartments (for example, ``mild,” ``moderate,” and ``severe” compartments), our model captures the continuous, time-varying relationship between individual immune-mediated symptom development and population-level transmission rate modification for symptomatic individuals. The proposed multi-scale framework is particularly important for diseases where symptomatic individuals can dynamically adjust their contact patterns based on illness severity, while asymptomatic individuals tend to maintain unchanged behavior despite potential transmission.



This study makes two interconnected contributions to multi-scale and behavioral epidemiological modeling:
\begin{enumerate}
    \renewcommand{\labelenumi}{(\alph{enumi})}
    \item We advance nested immuno-epidemiological modeling by linking within-host and population-level scales through a within-host model derived symptom score, which is better aligned to inform behavior change at a population-level than viral load.
    \item Through this novel immuno-epi model, we advance mechanistic modeling of behavior change during an outbreak through a distinct behavior-change pathway: transmission suppression driven by an individual's own infection-age clock, in contrast to typical population-level feedback (e.g., humans change behavior due to reported cases or deaths observed in the community) models in the literature \cite{pant2024behavior,pant2026paradox}.
\end{enumerate}

This study is organized as follows. In \Cref{sec2}, the immuno-epidemiology model is formulated. Specifically, in \Cref{sec2.1}, the within-host immunology model is presented, along with the local asymptotic stability result for its reproduction number. In \Cref{sec2.2}, the between-host transmission dynamics model is presented, along with the theoretical results on local asymptotic stability of its associated reproduction number and uniqueness of the endemic equilibrium. In \Cref{sec:sensitivity}, the impact of varying key model parameters is assessed. \Cref{sec:SEIR} qualitatively compares the outcomes of the multi-scale model to a standard SEIR model that does not account for behavior change and a modified SEIR model that explicitly does so. Finally, \Cref{sec: discussion} presents the discussion and conclusion.

\section{Model Formulation}\label{sec2}

We propose a {multiscale} model of infectious disease dynamics linking within-host viral dynamics and population-level transmission of the disease. The within-host model {(which is originally introduced in \cite{canini2011population}) tracks viral load, immune responses, and symptom {development, as described by a system of ordinary differential equations (ODEs)}. {We formulate an age-since-infection structured partial differential equation (PDE) model that incorporates the effects of symptom-based behavioral changes through the usage of linking functions. Specifically, symptomatic individuals in the population-level model change their behavior based on the severity of their symptoms, as inferred by the within-host model.

\subsection{Within-host dynamics}\label{sec2.1}

It is assumed that susceptible target cells ($T_S$) get infected by virus  ($V$) at a rate $
\phi$, upon which they become nonproductive but infected (eclipse) cells ($T_E$). These eclipse cells become productive infected cells ($T_I$) at a rate $\omega$. Productive infected cells die naturally at a rate $\delta_n$ or die due to natural killer cells ($K$) at a rate $\delta_k$. Virus titer is produced by productive infected cells at a rate $\frac{r}{1+\psi F}$, where cytokine-mediated inhibition reduces the baseline production rate $r$ in proportion to pro-inflammatory cytokine ($F$) levels, with $\psi$ representing the strength of this inhibitory effect.
Virus titer is cleared at a rate $c$. It is assumed that cytokines and natural killer cells production rates are $\nu_i$ and $
\nu_k$, respectively, while their clearance rates are  $\alpha$ and $
\eta$, respectively. Finally, it is assumed that the rate at which systemic symptoms ($Y$) appear is $\zeta$ while $\xi$ is the rate of symptom resolution. The resulting system of differential equations, adapted from \cite{canini2011population}, that describes the within-host dynamics is given by:



\begin{equation}\label{eq:within_host}
\left\{
\begin{aligned}
\frac{dT_S}{d\tau} &= -\phi T_S V, \\
\frac{dT_E}{d\tau} &= \phi T_S V - \omega T_E, \\
\frac{dT_I}{d\tau} &= \omega T_E - \delta_n T_I - \delta_k T_I K, \\
\frac{dV}{d\tau} &= \frac{r}{1 + \psi F} T_I - c V, \\
\frac{dF}{d\tau} &= \nu_i T_I - \alpha F, \\
\frac{dK}{d\tau} &= \nu_k F - \theta K, \\
\frac{dY}{d\tau} &= \zeta F - \xi Y,
\end{aligned}
\right.
\end{equation}



The within-host model \eqref{eq:within_host} has a line of {\it infection-free equilibria} (IFE), given by:
\begin{equation}\label{eq:IFEDFE}
\mathbb{E}_0 = (T_S,T_E,T_I,V, F, K, Y) = ({T}^\diamond_S,0,0,0,0,0,0),  
\end{equation}
where $T_S^\diamond$ is the number of target cells at the infection-free equilibria. However, in an immunological context, we take $T_S^\diamond$ to be the number of target cells before infection, i.e., $T_S^\diamond=T_S(0).$

Following the next generation operator method \cite{van2002reproduction,diekmann1990definition}, it follows that the associated non-negative matrix of new infection terms
($\mathbb{F}$) and the M-matrix of all linear transition terms ($\mathbb{V}$) are given, respectively, by 


\begin{equation*}
\mathbb{F} = \begin{bmatrix} 0 & 0 & \phi T_S^\diamond \\ 0 & 0 & 0 \\ 0 & r & 0 \end{bmatrix} \quad \text{and} \quad \mathbb{V} = \begin{bmatrix} \omega & 0 & 0 \\ -\omega & \delta_n & 0 \\ 0 & 0 & c \end{bmatrix}.
\end{equation*}

The {\it immunological reproduction number} ($\mathscr{R}_0$), which measures the number of secondary viral particles that one viral particle will produce in a population consisting entirely of susceptible target cells, is given by the spectral radius of $\mathbb{F}\mathbb{V}^{-1}$,

\begin{equation}\label{eq:R0withinhost}
\mathscr{R}_0 =  \underbrace{\Big(\phi T_S^\diamond\Big)}_{\text{infection rate at IFE}} \quad \times \underbrace{\left(\frac{1}{\delta_n}\right)}_{\substack{\text{average lifespan of} \\ \text{productive infected cells at IFE}}} \times \underbrace{r}_{\substack{\text{rate of virus production} \\ \text{per productive infected cells}}} \times \underbrace{\left(\frac{1}{c}\right).}_{\substack{\text{average lifespan} \\ \text{of virus titer}}}
\end{equation}



This immunological reproduction number ($\mathscr{R}_0$) is a threshold quantity. If $\mathscr{R}_0<1$, viral particles decay exponentially and the infection fails to establish. In fact, even when $\mathscr{R}_0>1$, the viral particles will eventually clear out (despite initial exponential growth) as the model is structured to describe a non-chronic disease. This eventual clearance of viral particles even when $\mathscr{R}_0>1$ occurs due to two biological mechanisms: (i) target cells (respiratory epithelial cells lining the airways) are not replenished during acute infection, as epithelial turnover occurs over 30-50 days \cite{crystal2008airway} while influenza-like illness typically resolves within 14 days, and (ii) immune-mediated reduction of viral production and viral elimination through intrinsic viral clearance, ensuring eventual elimination.

\begin{figure}[ht]
    \centering
\includegraphics[width=0.5\textwidth]{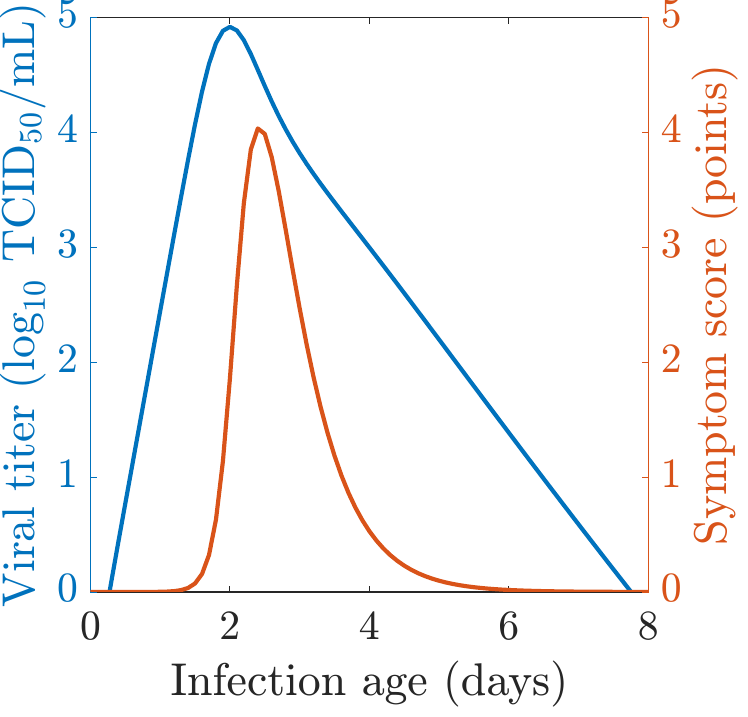}
\caption{Simulation of within-host model \eqref{eq:within_host} showing viral titer $\log_{10}(V(\tau))$ (blue, left axis) 
and symptom score $Y(\tau)$ (red, right axis) as functions of infection age $\tau$ 
(days since infection). This is a partial recreation of Figure 3 in \cite{canini2011population}. Parameter values used in the simulation are given in \Cref{tab:param_value}.}
    \label{fig:withinhost}
\end{figure}

\subsection{Population-level disease dynamics}\label{sec2.2}

We assume that the total population at time $t$, defined as $N(t)$, can be partitioned into compartments of susceptible, $S(t)$, exposed but not infectious, $E(t)$, symptomatically-infectious, ${I}(t)$, asymptomatically-infectious, $A(t)$, and recovered, $R(t)$, individuals, such that

$$N(t) = S(t) + E(t) + {I}(t) + A(t) + R(t),$$
where ${I}(t)=\int_0^\infty i(\tau,t) d\tau$ with $i(\tau,t)$ representing the density of symptomatically infectious individuals at time t with {infection age} $\tau$.



We assume that susceptible individuals acquire the disease at a rate $\lambda(\tau,t)$. After $1/\sigma$ duration, $\rho$ proportion of exposed but not yet infectious individuals became symptomatically infectious while the remaining $1-\rho$ became asymptomatically infectious. We assume that $i(\tau,t)$ recover at a rate $\gamma_i(\tau)$, while asymptomatic individuals recover at a rate $\gamma_A$ days. The function $\gamma_i(\tau)$ can be structured based upon the viral load concentration and the immune response of the host \cite{GulbudakRift,GulbudakDengue}. However, within our analyses we assume $\gamma_i(\tau)$ to be any arbitrary continuous differentiable function, {while for simulation purposes we assume it to be a fixed parameter. }  

We define the {\it force of infection}, denoted by $\lambda(t)$, as:

\begin{equation}\label{eq:foi}
 \lambda(t) = \beta \left[\dfrac{ \displaystyle\int_0^\infty \eta(Y(\tau)){i(\tau, t)} \, d\tau + A(t)}{N}\right],    
\end{equation}
where $\beta$ represents the baseline population-level transmission rate, defined as the product of the baseline contact rate and the baseline per-contact transmission probability. We define $\eta(Y(\tau))$, with $0 \leq \eta(Y(\tau)) \leq 1$, as a transmission modification function that captures how symptom severity affects the effective transmission rate of symptomatically infectious individuals. 
The transmission modification occurs through two distinct but coupled mechanisms: (i) symptomatic individuals reduce their contact patterns due to illness severity, social responsibility, or physical incapacity, and (ii) symptomatic individuals may adopt protective behaviors such as mask-wearing that reduce the per-contact transmission probability. For mathematical tractability, we model these effects jointly through the single function $\eta(Y(\tau))$, where decreasing values correspond to greater transmission reduction as symptom severity increases. Asymptomatic individuals maintain the baseline transmission rate $\beta$, since we assume they are unaware of their infectious status and thus do not modify their behavior patterns or adopt additional protective measures.

The initial number of individuals in state $X$ with $X = \{S,E,I,A,R\}$ is given by $X(0)$, with the initial total number of symptomatically infectious individuals expressed as $I(0) = \int_0^\infty i(\tau,0) d\tau$. The boundary condition $i(0,t) = \rho \sigma E(t)$ represents the flow of individuals transitioning from the exposed class to become newly symptomatically infectious at time $t$ (i.e., individuals with zero time since infection). The aggregate recovery rate of symptomatically infectious individuals is given by $\int_0^\infty \gamma_i(\tau)i(\tau,t) d\tau$.

Based on these assumptions, the population-level transmission dynamics of ILI are described by the following coupled ODE-PDE system:

\begin{equation}\label{eq:population_level}
\left\{
\begin{aligned}
\frac{dS}{dt} &= \Lambda -\beta \left( \displaystyle\int_0^\infty \eta(Y(\tau))i(\tau, t) \, d\tau + A(t) \right) \frac{S}{N} - \mu S, \\
\frac{dE}{dt} &= \beta \left( \displaystyle\int_0^\infty \eta(Y(\tau))i(\tau, t) \, d\tau + A(t) \right) \frac{S}{N} - (\sigma+\mu) E, \\
\frac{\partial i(\tau, t)}{\partial \tau} + \frac{\partial i(\tau,t)}{\partial t} &= -(\gamma_i(\tau)+\mu) i(\tau,t), \quad \quad  \quad  i(0,t) = \rho \sigma E(t) \\
\frac{dA}{dt} &= (1-\rho) \sigma E - (\gamma_A+\mu) A(t), \\
\frac{dR}{dt} &= \int_0^\infty \gamma_i(\tau) i(\tau, t) \, d\tau + \gamma_A A(t) - \mu R.
\end{aligned}
\right.
\end{equation}

\begin{table}[ht] 
  \begin{center}
    \begin{tikzpicture}[baseline]
      \node (table) {\normalsize
        \begin{tabular}{|c|p{11cm}|}
          \hline \hline
          \textbf{State Variable} & \textbf{Description}  \\
          \hline \hline
          $T_S(\tau)$ & Number of uninfected target epithelial cells at time $\tau$ since infection \\
          $T_E(\tau)$ & Number of nonproductive infected eclipse cells at time $\tau$ since infection \\
          $T_I(\tau)$ & Number of productive infected cells at time $\tau$ since infection \\
          $V(\tau)$ & Concentration of virus at time $\tau$ since infection \\
          $F(\tau)$ & Number of pro-inflammatory cytokines at time $\tau$ since infection \\
          $K(\tau)$ & Number of natural killer cells at time $\tau$ since infection \\
          $Y(\tau)$  & Systemic symptom score at time $\tau$ since infection \\
          \hline \hline
          $S(t)$ & Number of susceptible individuals at time $t$ \\
          $E(t)$ & Number of exposed individuals at time $t$ \\
          $i(\tau,t)$ & Density of symptomatically infectious individuals at time $t$ with age since infection $\tau$ \\
          $A(t)$ & Number of asymptomatically infectious individuals at time $t$ \\
          $R(t)$ & Number of recovered individuals at time $t$ \\
          \hline \hline
        \end{tabular}
      };
      \draw[thick, black, decorate, decoration={brace, amplitude=8pt, mirror}] 
        ([yshift=-21pt]table.north west) -- ([yshift=-105pt]table.north west) 
        node[midway, left=7pt] {\normalsize Within-host};
      \draw[thick, black, decorate, decoration={brace, amplitude=8pt, mirror}] 
        ([yshift=-107pt]table.north west) -- ([yshift=-180pt]table.north west) 
        node[midway, left=7pt] {\normalsize Population-level};
    \end{tikzpicture}
    \caption{Description of the state variables in the immuno-epidemiology behavior model \model.}
    \label{tab:state_variables}
  \end{center}
\end{table}

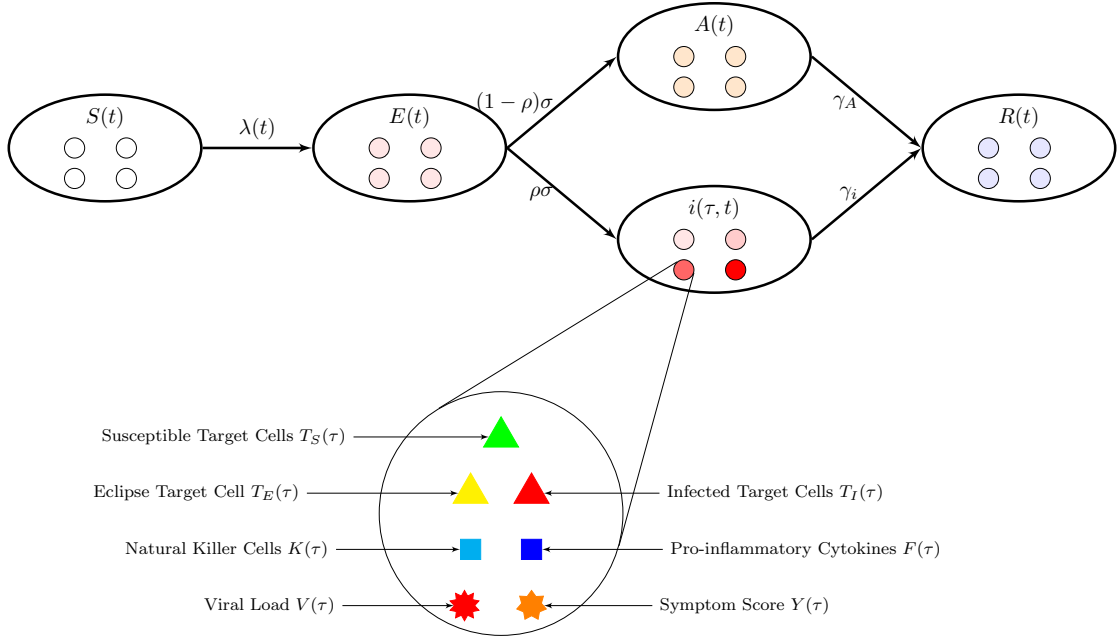
\begin{figure}
\centering
\resizebox{0.9\textwidth}{!}{
\begin{tikzpicture}[node distance=4cm, auto]
\node[circ] (s) at (0,0) {}; 
\node[circ] (m) at (5,0) {}; 
\node[circ] (a) at (10, 1.5) {}; 
\node[circ] (i) at (10, -1.5) {}; 
\node[circ] (r) at (15, 0) {}; 
\node (stext) at (0,.5) {$S(t)$};
\node (etext) at (5,.5) {$E(t)$};
\node (atext) at (10, 2) {$A(t)$};
\node (itext) at (10, -1) { $ i(\tau,t)$};
\node (rtext) at (15, 0.5) {$R(t)$};
   \node[sc] (s1) at (-.5,0) {};
   \node[sc] (s2) at (.35,0) {};
   \node[sc] (s3) at (-.5,-.5) {};
   \node[sc] (s4) at (.35,-.5) {};
   \node[mc] (m1) at (4.5,0) {};
   \node[mc] (m2) at (5.35,0) {};
   \node[mc] (m3) at (4.5,-.5){};
   \node[mc] (m4) at (5.35,-.5){};
   \node[rc] (r1) at (14.5,0) {};
   \node[rc] (r2) at (15.35,0) {};
   \node[rc] (r3) at (14.5,-.5) {};
   \node[rc] (r4) at (15.35,-.5) {};
   \node[ac] (a1) at (9.5,1.5) {};
   \node[ac] (a2) at (10.35,1.5) {};
   \node[ac] (a3) at (9.5, 1) {};
   \node[ac] (a5) at (10.35, 1) {};
   \node[draw,fill=red!10,circle] (i1) at (9.5,-1.5) {};
   \node[draw,fill=red!20,circle] (i2) at (10.35,-1.5) {};
   \node[draw,fill=red!60,circle] (i3) at (9.5,-2) {};
  \node[draw,fill=red,circle] (i4) at (10.35,-2) {};
    \draw [line] (s.east) -- (m.west) node [midway, above] {$\lambda(t)$};
    \draw [line] (m.east) -- (a.west) node [midway, left] {$(1-\rho)\sigma$};
    \draw [line] (m.east) -- (i.west) node [midway, left]  {$\rho \sigma$};
    \draw [line] (i.east) -- (r.west) node [midway, left]  {$\gamma_i$};
    \draw [line] (a.east) -- (r.west) node [midway, left]  {$\gamma_A$};
  \node[draw,circle, minimum size=4cm] (whd) at (6.5,-6) {};
   \draw (tangent cs:node=i3,point={(whd.west)},solution=2) -- (tangent cs:node=whd,point={(i3.west)});
   \draw (tangent cs:node=i3,point={(whd.east)},solution=1) -- (tangent cs:node=whd,point={(i3.east)},solution=2);
   \node[draw=green,fill=green,regular polygon,regular polygon sides=3,align=center] (a) at ([yshift=-.75cm]whd.north){};
   \node[draw=yellow,fill=yellow,regular polygon,regular polygon sides=3,align=center] (ia) at ([xshift=-.5cm,yshift=-.75cm]a.south){};
   \node[draw=red,fill=red,regular polygon,regular polygon sides=3,align=center] (tactive) at ([xshift=.5cm,yshift=-.75cm]a.south){};
    \node[draw=cyan,fill=cyan,regular polygon,regular polygon sides=4,align=center] (killer) at ([xshift=-1 cm,yshift=-.75cm]tactive.south){};
\node[draw=blue,fill=blue,regular polygon,regular polygon sides=4,align=center] (Fheal) at ([xshift=0 cm,yshift=-.75cm]tactive.south){};
   \node[draw = red,fill=red, ellipse,star,star points=10, align=center] (p) at ([xshift= -0.1cm,yshift =-0.75cm]killer.south){};
   \node[draw = orange,fill=orange, ellipse,star,star points=7, align=center] (symp) at ([xshift= 1cm,yshift =-0.75cm]killer.south){};
   \node[draw=none,fill=none,left of=a,xshift=-.55cm] (atext) {\footnotesize Susceptible Target Cells  $T_{S}(\tau)$};
   \node[draw=none,fill=none,left of=ia,xshift=-.5cm] (iatext) {\footnotesize  Eclipse Target Cell $T_{E}(\tau)$};
   \node[draw=none,fill=none,right of=tactive] (infecttext) {\footnotesize Infected Target Cells $T_{I}(\tau)$};
   \node[draw=none,fill=none,left of=p,xshift=0.8cm] (ptext) {\footnotesize Viral Load $V(\tau)$};
   \node[draw=none,fill=none,left of=killer,xshift=0 cm] (killertext) {\footnotesize Natural Killer Cells $K(\tau)$};
   \node[draw=none,fill=none,right of=Fheal,xshift=0.5 cm] (Fhealtext) {\footnotesize Pro-inflammatory Cytokines $F(\tau)$};
   \node[draw=none,fill=none,right of=symp,xshift=-0.5 cm] (symptext) {\footnotesize Symptom Score $Y(\tau)$};
  \draw[-latex'] (atext.east) -- (a.west);
   \draw[-latex'] (iatext.east) -- (ia.west);
   \draw[-latex'] (ptext.east) -- (p.west);
   \draw[-latex'] (killertext.east) -- (killer.west);
   \draw[-latex'] (infecttext.west) --(tactive.east);
   \draw[-latex'] (Fhealtext.west) -- (Fheal.east); 
   \draw[-latex'] (symptext.west) -- (symp.east);
\end{tikzpicture}  
}
\caption{Flowchart of immuno-epidemiology behavior model \model. Arrows depicting birth and death {in the population-level model} are omitted for simplicity. {The circled region is a schematic of the within-host model that depicts that each symptomatic individual in $i(\tau,t)$ has an associated within-host dynamics that evolve with time since infection $\tau$. In contrast, asymptomatic individuals ($A$) are not assigned within-host dynamics.} Description of state variables and parameters of the model are given in Tables \ref{tab:state_variables} and \ref{tab:param_value}, respectively. The expression for force of infection $\lambda(t)$ is given by \eqref{eq:foi}.}
\end{figure}

\begin{table}[ht]
  \begin{center}
    \small
    \begin{tabular}{|c|c|p{6.5cm}|c|c|}
      \hline \hline
      & \textbf{Parameter} & \textbf{Description} & \textbf{Value} & \textbf{Source} \\
      \hline \hline
      \multirow{13}{*}{\rotatebox{90}{Within-host}}
      & \(\phi\) & Infection rate of target cells & $3.0 \times 10^{-5}$ (TCID$_{50}$/ml)$^{-1}$ day$^{-1}$ & \cite{canini2011population} \\
      & \(\omega\) & Transition rate from \(T_E\) to \(T_I\) & 2.8 day$^{-1}$ & \cite{canini2011population} \\
      & \(\delta_n\) & Mortality rate of \(T_I\) & 1.56 day$^{-1}$ & \cite{canini2011population} \\
      & \(\delta_k\) & K-mediated killing rate & $3.8 \times 10^{-6}$ ([K]$\times$day)$^{-1}$ & \cite{canini2011population} \\
      & \(r\) & Rate of virus production by \(T_I\) & 0.043 [V]$\times([T_I]\times$day$)^{-1}$ & \cite{canini2011population} \\
      & \(\psi\) & Effect of viral growth due to the presence of cytokines \(F\) & $2.4 \times 10^{-7}$ {[F]$^{-1}$} & \cite{canini2011population} \\
      & \(c\) & Virus clearance rate & 7.1 day$^{-1}$ & \cite{canini2011population} \\
      & $\nu_{i}$ & Activation of cytokine production due to the presence of \(T_I\) & 1 [F]$\times([T_I]\times$day$)^{-1}$ & \cite{canini2011population} \\
      & \(\alpha\) & Cytokine clearance rate & 1.82 day$^{-1}$ & \cite{canini2011population} \\
      & \(\nu_{{k}}\) & Activation of killer cells & 1 [K]$\times([F]\times$day$)^{-1}$ & \cite{canini2011population} \\
      & \(\theta\) & $K$ cell clearance rate & 0.061 day$^{-1}$ & \cite{canini2011population} \\
      & \(\zeta\) & Symptom activation rate & $2.2 \times 10^{-6}$ symptom score points [F]$^{-1}$ & \cite{canini2011population} \\
      & \(\xi\) & Symptom resolution rate & 6.1 day$^{-1}$ & \cite{canini2011population} \\
      \hline \hline
      \multirow{9}{*}{\rotatebox{90}{Population-level}}
      & $\Lambda$ & Recruitment rate of susceptible individuals & 3.5215 day$^{-1}$ & Assumed \\
      & \(\beta\) & Baseline population-level infection rate & 0.4 day$^{-1}$ & Assumed \\
      & \(\rho\) & Fraction of exposed individuals who become symptomatic & 0.8 (dimensionless) & Assumed \\
      & \(\sigma\) & {Transition rate from E to $i$ or $A$} &  $1/1.6$ day$^{-1}$ & {\cite{cori2012estimating}} \\
      & \(\gamma_i(\tau), \gamma_A\) & Recovery rates for \(i\) and \(A\) & $\gamma_i=\gamma_A=1/5$ day$^{-1}$ & Assumed \\
      & $\mu$ & Natural mortality rate of individuals & $1/(77.8\times365)$ day$^{-1}$ & \cite{pant2024hetro} \\
      & $\eta(\tau)$ & Reduction in effective transmission (relative to baseline $\beta$) due to symptom-driven contact/behavior change & --- & --- \\
      \hline \hline
    \end{tabular}
    \caption{Description, values, and sources of the parameters in the immuno-epidemiology behavior model \model. The within-host model \eqref{eq:within_host} is initialized with $T_S(0)=4\times10^8$ target cells, $V(0)=0.5$ TCID$_{50}$/ml of virus, and all remaining state variables ($T_E,T_I,F,K,Y$) set to zero \cite{canini2011population}. The between-host model is initialized with $N(0)=100{,}000$, $S(0)=99{,}995$, $E(0)=5$, and $I(0)=A(0)=R(0)=0$.}
    \label{tab:param_value}
  \end{center}
\end{table}

\subsection{Linking functions connecting the two scales}\label{sec:linking}

The connection between within-host and between-host dynamics in immuno-epidemiological 
models is typically achieved through a {\it linking function} that specifies how within-host 
state variables influence population-level parameter.
Following the nested approach of Gilchrist and Sasaki \cite{gilchrist2002modeling}, 
most immuno-epidemiological models assume that the population-level transmission rate 
is a function of viral load $V(\tau)$ at the within-host scale \cite{martcheva2015introduction,gilchrist2002modeling}. Common functional forms for this viral-load-based linking include 
linear \cite{gilchrist2002modeling,saldana2025multiscale, martcheva2015introduction}, exponential \cite{saldana2025multiscale, martcheva2015introduction}, and Hill functions  \cite{fraser2007variation,atkins2023immuno,saldana2025multiscale, martcheva2015introduction}.

The connection between viral load and transmission 
for influenza-like illness is, however, non-trivial since viral load dynamics often misalign with 
symptom severity \cite{canini2011population,challenger2022modelling,he2020temporal,argyropoulos2020association}.
Furthermore, symptomatic individuals are known to modify their behavior when ill, 
reducing contact patterns due to illness severity, unwillingness to infect others, 
or simple inability to maintain normal routines \cite{mclean2017school}.
{Focused on only understanding symptom driven behavior change, we propose} an alternative linking paradigm where the transmission modification 
function for symptomatic individuals $\eta$ depends on systemic symptom scores $Y(\tau)$ generated by the 
within-host immune dynamics model \eqref{eq:within_host}.
%
For simplicity, we connect the two scales using a standard exponential decay linking function that decays with increasing symptom score and takes the following form \cite{saldana2025multiscale, martcheva2015introduction}:
\begin{equation}\label{eq:exponential}
    \eta(Y) = e^{-a Y},
\end{equation}
where parameter $a>0$ controls the rate at which increasing symptom scores 
reduce contact rates. Henceforth we define parameter $a$ as the {\it behavior change parameter}. We later examine other commonly used linking functions to assess whether this choice qualitatively affects our results (see Section \ref{sec:SEIR}).


Apart from the transmission modification function $\eta(\tau)$, the recovery rate of 
symptomatically infectious individuals $\gamma_i(\tau)$ in the population-level model \eqref{eq:population_level} is also informed by the within-host scale. For simplicity we assume that $\gamma_i(\tau)$ is a fixed parameter, i.e.,  $\gamma_i(\tau)= \gamma_i$ is fixed. 

\subsection{Equilibria and threshold analysis of the population-level model}
A unique {\it disease-free equilibrium} (DFE) of the epidemiological model \eqref{eq:population_level}, obtained by setting 
derivatives with respect to $t$ in \eqref{eq:population_level} equal to zero, is given by:

\begin{equation}\label{eq:DFEpop}
\mathcal{E}_0 = ({S}^*,{E}^*,{i}^*(\tau),{A}^*,{R}^*)=\left(\dfrac{\Lambda}{\mu},0,0,0,0\right).
\end{equation}

We remark that the total number of symptomatically infectious individuals at DFE ($I^*$) is given by: ${I}^*=\int_0^\infty i^*(\tau) d\tau$.

The {\it {control} reproduction number} ($\mathcal{R}_C$), defined as the average number of new cases generated by a typical infectious individual {(in the presence of the symptom-driven behavioral change)} when introduced to a population almost entirely consisting of a susceptible population, 
of the epidemiological model \eqref{eq:population_level} can be expressed as the sum of two constituent reproduction numbers:

\begin{equation}\label{eq:R0pop}
\mathcal{R}_C = \mathcal{R}_{C,I} + \mathcal{R}_{C,A},
\end{equation}
where

\begin{equation}\label{eq:R_cons}
\begin{aligned}
\mathcal{R}_{C,I} &=
{\int_0^\infty 
\overbrace{\beta \eta(Y(\tau)) \frac{{S}^*}{{N}^*}}^{\substack{\text{Effective initial} \\ \text{transmission rate by} \\ \text{ individuals with infection age $\tau$}}} 
\overbrace{\exp\left( - \int_0^\tau \big(\gamma_I(({\widehat{\tau}})) + \mu\big) \,d{\widehat{\tau}} \right)}^{\substack{\text{Probability of remaining} \\ \text{infectious at time } \tau}}d\tau}_{\substack{\text{} \\ \text{}}} 
\quad \times 
\underbrace{\frac{\rho\sigma}{\sigma + \mu}}_{\substack{\text{Probability of surviving} \\ \text{exposed class and becoming} \\ \text{symptomatically-infectious}}},\\
\mathcal{R}_{C,A} &=
\underbrace{\beta \frac{{S}^*}{{N}^*}}_{\substack{\text{Effective initial} \\ \text{transmission rate by}\\\text{asymptomatic individuals}}} \times
\underbrace{\frac{(1 - \rho)\sigma}{\sigma + \mu}}_{\substack{\text{Probability of surviving} \\ \text{exposed class and becoming}\\\text{asymptomatically-infectious}}} \times
\underbrace{\frac{1}{\gamma_A + \mu}.}_{\substack{\text{Average duration} \\ \text{in asymptomatic class}}}
\end{aligned}
\end{equation}



The constituent control reproduction numbers $\mathcal{R}_{C,I}$ and $\mathcal{R}_{C,A}$ are the number of secondary cases generated respectively by a symptomatically-infectious and an asymptomatically-infectious individual when introduced to a population consisting of almost entirely susceptible population. It is convenient to define
\begin{equation}
\pi(\tau)=\exp{\left(-\int_0^\tau\big(\gamma_i({{\widehat{\tau}}})+\mu\big) d{\widehat{\tau}}\right)},
\end{equation}

where $\pi(\tau)$ is the {survival function representing the} probability of a symptomatic infectious individual remaining infectious $\tau$ time units after infection. Furthermore, we define 
\begin{equation}
    \Pi = \int_0^\infty \pi(\tau) \,d\tau,
\end{equation}
where $\Pi$ is the average duration of symptomatic infection.
Thus, in compact form the $\mathcal{R}_C$ can be written as:
\begin{equation}
\frac{\beta}{\sigma+\mu}\left[\rho\sigma\int_0^\infty\eta(\tau)\pi(\tau)d\tau+\frac{(1-\rho)\sigma}{\gamma_A+\mu}\right].
\end{equation}

{The \textit{basic reproduction number} ($\mathcal{R}_0$), defined as the average number of new
cases generated by a typical infectious individual when introduced to a population almost
entirely consisting of a susceptible population, in the \textit{absence} of symptom-driven
behavioral modification of transmission (i.e., $\eta(\tau) = 1$), of the model \modelsp is given by}
\begin{equation*}
    \mathcal{R}_0 = \mathcal{R}_C\big|_{\eta = 1}.
\end{equation*}
{That is, $\mathcal{R}_0$ recovers the reproduction number of the corresponding
intervention-free model in which symptomatic individuals do not modify their contact
patterns.}
\vspace{0.5cm}

\begin{theorem}\label{thm:laspop}
The disease-free equilibrium (DFE) $\mathcal{E}_0$  is locally asymptotically stable (LAS) if $\mathcal{R}_C < 1$ and unstable if $\mathcal{R}_C > 1$.
\end{theorem}
The proof of \Cref{thm:laspop} is given in \Cref{app:LAS}.

\vspace{1mm}

\begin{theorem}\label{thm:endemicUNIQUE}
    If $\mathcal{R}_C>1$, then the population-level model \eqref{eq:population_level} has a unique endemic equilibrium. 
\end{theorem}
The proof of \Cref{thm:endemicUNIQUE} is given in \Cref{app:endemicunique}.


\section{Effects of Behavioral and Symptom Severity Parameters on Outbreak Dynamics}\label{sec:sensitivity}

The novelty of this study lies in inducing behavior change through a within-host-derived symptom score, so our
simulations focus on the parameters governing this pathway rather than revisiting within-host dynamics in
isolation or standard symptomatic/asymptomatic SEIR extensions, both of which are already well studied
elsewhere (see, e.g., Canini and Carrat \cite{canini2011population} for the former, and \cite{gumel2021primer,pant2024hetro} for the latter).
Specifically, we examine the behavior change parameter $a$, which governs the strength of the contact-reduction
function $\eta(\tau)$, and the symptom-ramping parameter $\zeta$, which governs how sharply the within-host
model translates immune activity into symptom severity and, in turn, into behavioral response. 
Because behavior change in our model acts exclusively on symptomatic individuals, its population-level impact
is inherently mediated by how many infections become symptomatic in the first place. We therefore also
examine how varying the proportion of exposed individuals who become symptomatic, $\rho$, impacts disease
dynamics.

The numerical scheme used to solve the multiscale ODE-PDE model \modelsp is given in Appendix
\ref{numerical_scheme}. Unless otherwise stated in the caption of a figure, the parameters used in the
simulations presented in this study are tabulated in \Cref{tab:param_value}. For simplicity, we assume that
the average recovery time is the same for both symptomatic and asymptomatic individuals.

\Cref{fig:eta_total_infectious} shows how increasing the behavior change parameter $a$ shapes the number of  infectious
individuals ($I+A$). As parameter $a$ increases, symptomatic individuals reduce contact more sharply (\Cref{fig:eta_total_infectious}(a)), which translates into a delayed, flattened epidemic peak at the
population level (\Cref{fig:eta_total_infectious}(b)).

\Cref{fig:zeta_eta_total} instead varies the symptom-ramping parameter $\zeta$, which controls how sharply the
within-host model translates immune activity into symptom severity. Larger $\zeta$ produces higher, more
pronounced symptom scores (\Cref{fig:zeta_eta_total}(a)). Higher symptom trajectory drives a
correspondingly deeper reduction in contact (\Cref{fig:zeta_eta_total}(b)), and in turn a smaller, delayed
epidemic peak (\Cref{fig:zeta_eta_total}(c)). Together, Figures \ref{fig:eta_total_infectious} and \ref{fig:zeta_eta_total} show that behavioral suppression
of transmission can be driven by a stronger behavioral response ($a$) or by more severe symptoms
($\zeta$), acting through the same underlying mechanism: reduced contact during the symptomatic period.

\Cref{fig:rho} examines the impact of change in the proportion of exposed individuals who become symptomatic ($\rho$) on disease dynamics. Increasing $\rho$ modestly delays and lowers the epidemic peak, as a larger share of infections shift
into the behaviorally-suppressed symptomatic class rather than the asymptomatic class where individuals do not change behavior
(\Cref{fig:rho}(a)). {Correspondingly, as $\rho \in (0,1)$ increases, the symptomatic contribution to the control reproduction number increases relative to that of asymptomatic individuals (\Cref{fig:rho}(b)).}
{Notably, however, the overall reproduction number $\mathcal{R}_C$ does not change dramatically across this wide range of $\rho$, suggesting that
symptom-driven behavior change alone, may have a comparatively modest effect
on the population's threshold transmission potential.} 

\begin{figure}[ht]
    \centering
\includegraphics[width=1\textwidth]{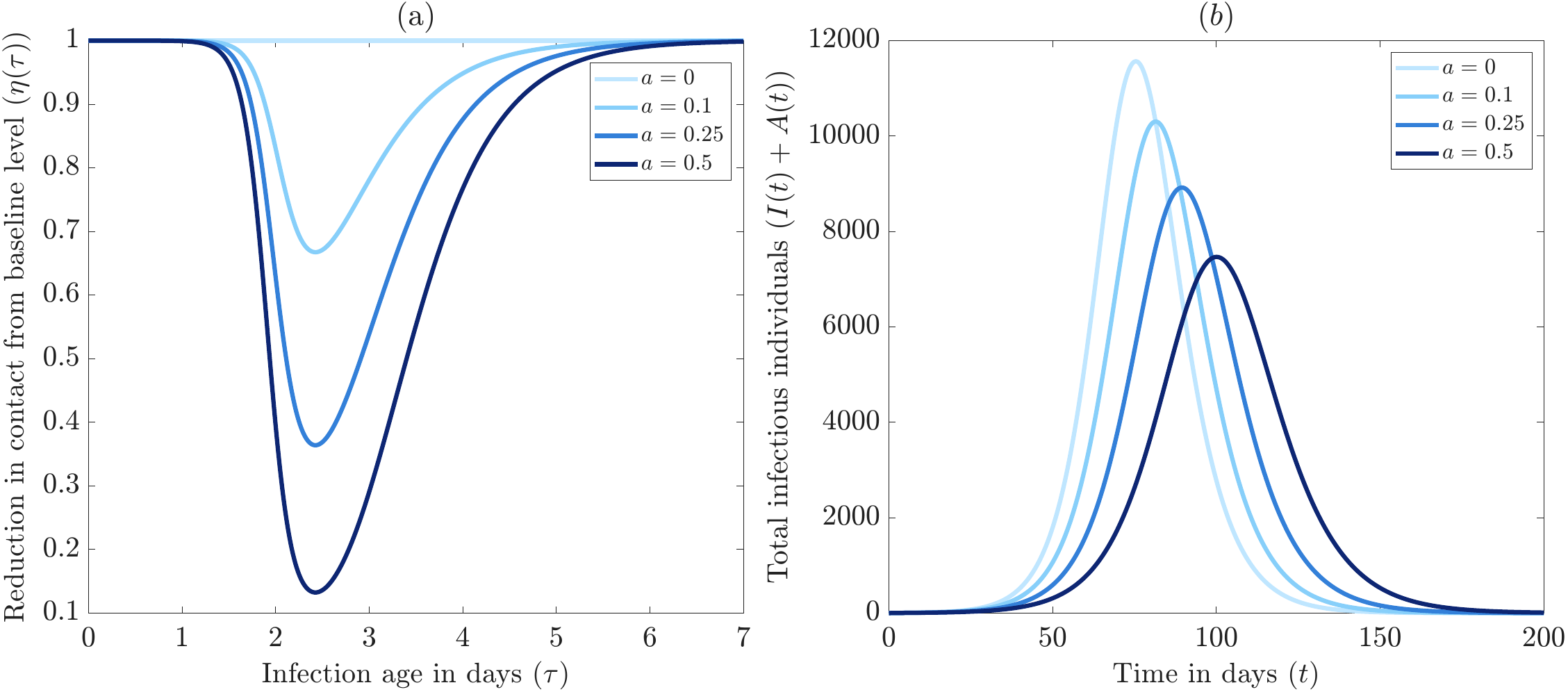}
\caption{Effect of the symptom-driven reduction in contact $\eta(\tau) = e^{-aY(\tau)}$ on epidemic dynamics, for varying values of the behavior change parameter $a$. \textbf{(a)} $\eta(\tau)$ as a function of infection age $\tau$, showing how larger $a$ produces a deeper and more prolonged reduction in contact during the symptomatic period. \textbf{(b)} Corresponding population-level impact on 
the total infectious population, $I(t)+A(t)$, over time $t$.
Increasing $a$ delays and flattens the epidemic peak, illustrating how stronger behavioral modulation to one's symptoms at the individual level dampens outbreak severity at the population level. Parameter values used in the simulation are given in Table~\ref{tab:param_value}.}
\label{fig:eta_total_infectious}
\end{figure}

\begin{figure}[ht]
    \centering
    \includegraphics[width=1\textwidth]{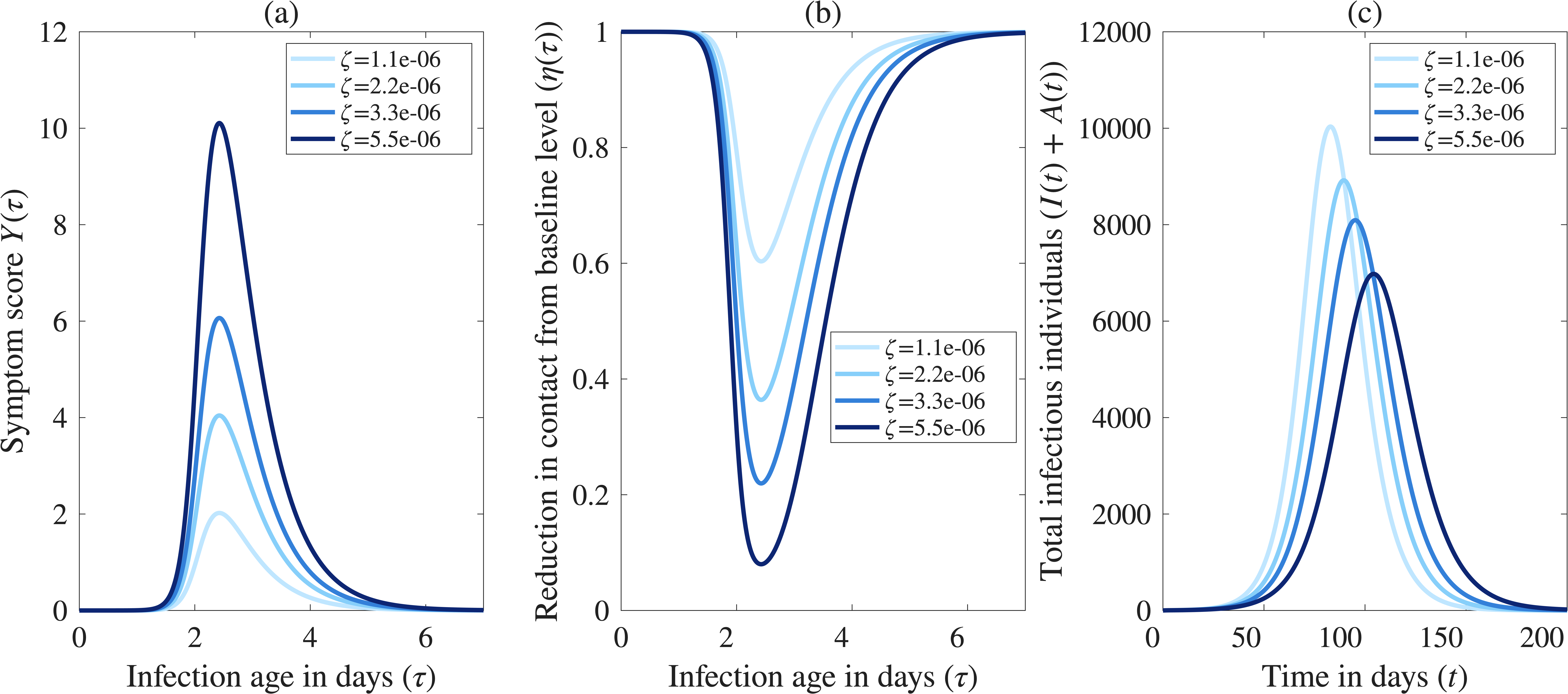}
\caption{Effect of varying the symptom-ramping parameter $\zeta$ on within-host symptom dynamics and between-host epidemic outcomes. (a) Within-host symptom score, $Y(\tau)$, as a function of infection age $\tau$ for increasing values of $\zeta$; larger $\zeta$ produces a higher, more pronounced symptom peak. (b) Corresponding contact-reduction function $\eta(\tau) = e^{-aY(\tau)}$, showing that higher symptom scores translate into a deeper and more prolonged reduction in contact rates during peak infection age. (c) Total infectious individuals, $I(t)+A(t)$, over epidemic time $t$ under each $\zeta$ scenario. The behavior parameter $a$ is fixed to 0.25, while the remaining parameters are as in \Cref{tab:param_value}.}
    \label{fig:zeta_eta_total}
\end{figure}

\begin{figure}[ht]
    \centering
    \includegraphics[width=1\textwidth]{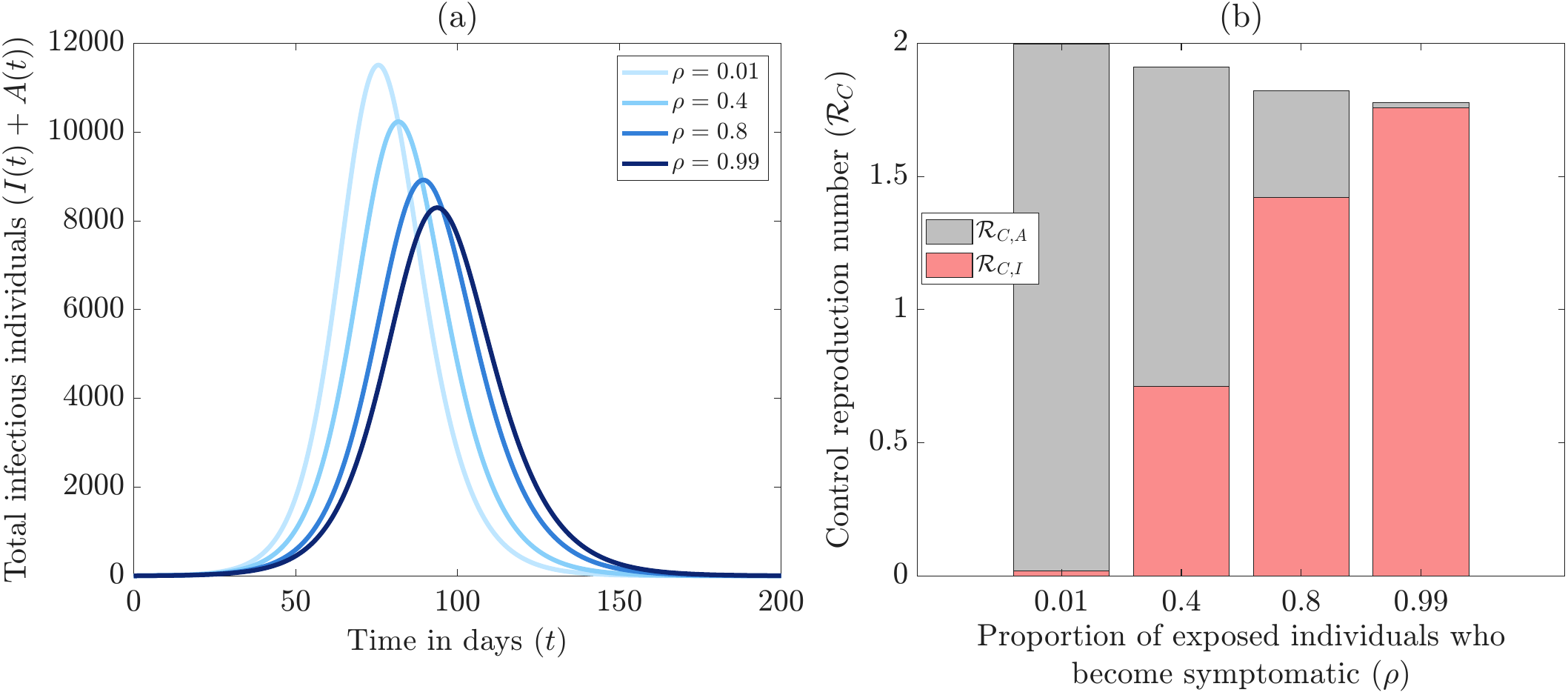}
\caption{Effect of varying the proportion of exposed individuals who become symptomatic, $\rho$, on epidemic dynamics and the basic reproduction number. (a) Total infectious individuals $I(t)+A(t)$ over epidemic time $t$ for four values of $\rho$. Increasing $\rho$ delays and lowers the epidemic peak, as a larger fraction of infections shift into the symptomatic class, where transmission is reduced by symptom-driven behavior change. (b) Corresponding contributions to the basic reproduction number, $\mathcal{R}_C=\mathcal{R}_{C,i}+\mathcal{R}_{C,A}$, decomposed into the symptomatic ($\mathcal{R}_{C,i}$) and asymptomatic ($\mathcal{R}_{C,A}$) components. As $\rho$ increases, the symptomatic contribution $\mathcal{R}_{C,i}$ grows while the asymptomatic contribution $\mathcal{R}_{C,A}$ shrinks. Since behaviorally-suppressed symptomatic transmission ($\eta(\tau)<1$) contributes less per capita than unmodified asymptomatic transmission, the total $\mathcal{R}_C$ declines as a larger share of infections become symptomatic. The behavior parameter $a$ is fixed to 0.25, while the remaining parameters are as in \Cref{tab:param_value}.}
    \label{fig:rho}
\end{figure}

\section{Comparison to Population-level SEIR Models}\label{sec:SEIR}


Various studies have shown that models {accounting} for population-level behavior change due to population-level reporting of disease (e.g., number of cases, hospitlization, or deaths) can generate qualitatively distinctive signatures such as slower post-peak decay, plateaus, or shoulders in
the epidemic curve \cite{weitz2020awareness, berestycki2021plateaus}. For instance, Pant et al. \cite{pant2026paradox} show that modifying the SEIR model to account for behavior change—under the assumption that increased disease activity leads to decreased transmission, and vice versa—results in heightened peak suppression along with a slower, more elongated epidemic decline as the strength of the behavioral response increases (an output of a model adapted from \cite{pant2026paradox} is given in Appendix \ref{SEIRb} for demonstration purposes). In contrast (as evident from Figures \ref{fig:eta_total_infectious}, \ref{fig:zeta_eta_total}, and \ref{fig:rho}), the shape of the trajectory of the number of infectious individuals generated through our multi-scale model \modelsp is akin to that of a standard population-level transmission model that does not account for behavior change. In fact, under the assumption $\gamma_i = \gamma_A$ (this assumption helps us isolate the effect of inducing behavior change in mult-scale mechanism), it turns out that the number of infectious individuals generated by the multiscale model (i.e., $I(t) +A(t)$) can be captured well by the infectious state variable of a standard SEIR model (i.e., $I(t)$ in SEIR). \Cref{fig:relative_error} demonstrates this by showing the relative error is ``reasonably small" even for large values of behavior parameter.

This qualitative feature of multi-scale model generating SEIR-like curves persists across other commonly used linking functions (specially Hill-type and linear) found in the literature. The precise functional forms of the  linking functions considered, their shape as a function of $\tau$ with respect to changes in the behavior-change parameter $a$, and the resulting consequences for the epidemic trajectory are given in \Cref{fig:linking_functions_comparison} in Appendix \ref{app:linking}. Moreover, using a linking function that zeros out contacts once symptoms exceed a certain threshold, we show that even when such extreme choice of linking functions are made SEIR-like curves can still be an output of the multiscale model (see \Cref{fig:isolation_linking}).

\begin{figure}[htbp]
    \centering
    \includegraphics[width=0.55\textwidth]{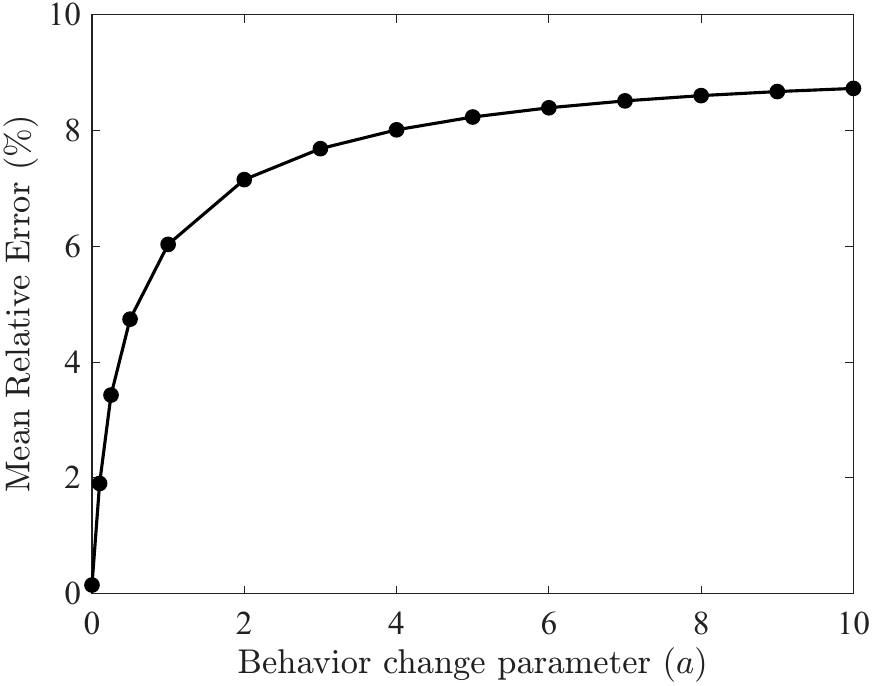}
\caption{Mean relative error between the true data generated from the multiscale model (i.e.,  $I(t)+A(t)$) and the SEIR model's  $I(t)$ fitted to it. For simplicity, only $\beta$ is fitted, while the remaining parameters $\gamma$ and $\sigma$, along with the initial conditions, are fixed to their true values in the SEIR model. For each value of $a$, the relative error is computed at every time point as $\mathrm{RE}(t) = \left|\frac{y_{\text{true}}(t) - y_{\text{fit}}(t)}{y_{\text{true}}(t)}\right| \times 100\%$, 
then averaged over all time points to give a single mean relative error per $a$ value. Circles indicate the sampled values of parameter $a$ at which the model was fit and the relative error computed, and the connecting line is a linear interpolation drawn for visual appeal.}
    \label{fig:relative_error}
\end{figure}

\newpage
\section{Discussion}\label{sec: discussion}

This study advances two related areas: (i) immuno-epidemiological modeling and (ii) behavioral epidemiology. Classical immuno-epidemiological models link the two scales through viral load, aiming to capture how infectivity changes over the course of infection. However, viral load is often misaligned with symptoms \cite{he2020temporal,canini2011population}, the very thing that drives behavior change, leaving this classical way of connecting the two scales inadequate for capturing behavior change. Meanwhile, typical population-level behavioral epidemiology models account for behavior change through population-level observations (e.g., reported cases, hospitalizations, or deaths) \cite{pant2024behavior,pant2026paradox,oveson2025modeling,lejeune2024mathematical}, largely ignoring illness-driven behavior change, a prominent mechanism of behavior change in many infectious diseases. To address this gap, we built a novel immuno-epi behavioral multi-scale model in which symptomatic individuals change their behavior based on a symptom score generated by a within-host immune model. Across several choices of linking function, we found that, {for the symptom-score trajectory generated by the within-host model considered here,} stronger symptom-driven behavior change delays and flattens the epidemic peak, but does not qualitatively alter the shape of the epidemic trajectory relative to a standard SEIR model without behavior change. In fact, when using an exponential linking function, we showed that the total number of infectious individuals generated by the multi-scale model can be closely reproduced by a standard SEIR model in which only the transmission rate is fitted.

Our main contribution to the immuno-epidemiological literature is the choice of linking variable itself: symptom score rather than viral load. Importantly, the study also shows that different functional forms of the linking function (e.g., exponential, Hill-type, and linear) can lead to qualitatively similar but quantitatively different results. This result aligns with the findings of another multiscale study \cite{gulbudak2022sensitivity}, where authors explored how different choices in linking functions (dependent on viral load and immune response) used in a immuno-epidemiological vector-host model impacted local sensitivity analysis results. Under the constraint of the same within-host dynamics and reproduction number at baseline, they found no qualitative difference in the resulting sensitivity indices or in the final disease abundance sizes across the different linking function choices.
This raises a broader question about the uncertainty and/or error introduced into complex multi-scale models by an incorrect choice of linking function — whether that means linking to the wrong within-host quantity, or assuming the wrong functional form. This mirrors a similar issue in population-level models, where different choices of force of infection (e.g., mass action, standard incidence, or Holling-type) can shape different dynamics, and, perhaps more relevantly, in behavioral epidemiology models, where different functional forms for behavior change produce a similar pattern. For instance, Pant et al. \cite{pant2026paradox} showed that different functional forms of transmission rate in a population-level behavioral model can produce qualitatively similar but quantitatively different results.

Our main contribution to the behavioral epidemiology literature is introducing behavior change that is generated through the within-host scale of a multi-scale model. This is in contrast to typical behavioral epidemiology models, which are largely built at the population level alone, without a multi-scale component — a gap that has recently been highlighted as a major barrier to understanding individual decision-making during infectious disease outbreaks \cite{saad2026interactions}. Importantly, this contrasts with recent population-level behavioral epidemiology models, which show that population-level feedback (i.e., individuals reacting to observed cases, hospitalizations, or deaths) can produce plateaus, shoulders, or an elongated epidemic decline \cite{weitz2020awareness, berestycki2021plateaus}. For instance, Pant et al. \cite{pant2026paradox} show that in a population-level model with such a feedback loop, stronger behavioral response leads to an elongated epidemic decline. This phenomenon is also present in trajectories generated by a group-structured behavioral model \cite{pant2024behavior}. Our illness-driven mechanism, by contrast, produces standard SEIR-like infectious trajectories. This is perhaps because, unlike population-level feedback models (where behavior change is induced by modifying the time-varying transmission rate $\beta(t)$ in response to changes in an observed aggregate outcome), in our multi-scale model behavior change is induced in symptomatic individuals by modifying their own transmission rate as a function of the severity of the symptoms they experience, an observable that evolves on a much shorter timescale than the outbreak and its associated aggregate outcome. 
This suggests that illness-driven behavior change of this form may be a comparatively weaker driver of transmission suppression at the aggregate scale.

This study has several limitations. First, although the model adds novelty by accounting for contact reduction as a function of a within-host-derived symptom score, it implicitly assumes that the probability of transmission per contact remains constant over time.  Second, although we study variability in symptom presentation by varying the symptom-activation parameter, we do not explore variability arising from viral load dynamics, the symptom clearance rate, or alternative symptom-score trajectories (e.g., bimodal symptom curves, or curves reflecting chronic conditions in which symptoms persist indefinite. {It is plausible that different choices of within-host dynamics, which give rise to different symptom score trajectories, might lead to population-level trajectories that deviate from standard SEIR-like curves.} Third, this study does not include data fitting, identifiability analysis (structural or practical), or an assessment of how uncertainty propagates from the within-host to the population scale. Future studies could address these gaps.

Overall, this study uses a within-host-model derived symptom score to inform behavior change at a population level, a mechanistically distinct way of inducing behavior change in behavioral epidemiological models, relative to population-feedback-driven approaches. {For the choices of symptom score curves considered,} we show that this immuno-epi behavioral model produces infection curves that closely resemble those of a standard SEIR model, underscoring that the mechanism generating an epidemic curve cannot always be inferred from its shape alone. Because several different choices of linking function all produced qualitatively SEIR-like curves, {it remains unclear if different choice of within-host dynamics and linking functions can cascade into population-level trajectories different than those produced by a standard SEIR-like model. Furthermore, its important to ask which additional data sets need to be observed to know if population-level disease trajectories are better described by a standard SEIR-like model that does not explicitly account for behavior change or by a multi-scale behavioral model.} Multi-scale models are typically fitted to within-host data (e.g., viral load) and population-level disease data (e.g., case counts), but whether {observing} additional novel data sources, even sparsely sampled \cite{pant2026uncovering}, (e.g., symptom-severity and corresponding behavioral survey data) could help identify the appropriate linking function remains an important open problem.


\paragraph{Acknowledgments.} Support for this research was provided to BP by an AMS-Simons Travel Grant. 
BP and HD would like to thank Dana Kristine Pasquale  and Jorge X. Velasco-Hernandez for organizing the online conference ``Building Teams to Build Better Epidemiological Models: Balancing Participation from Mathematical and Social, Behavioral, and Economic Sciences" at Duke University, which was instrumental in connecting BP and HD and ultimately seeding the initial idea for this project.

\paragraph{Generative AI Usage Statement.} During the preparation of this work, the authors used Claude (Anthropic), ChatGPT (OpenAI), and Gemini (Google) for the purposes of proofreading and improving the language and readability of the manuscript.  Additionally, the authors developed a master simulation file implementing the numerical scheme described in Appendix \ref{numerical_scheme}, and AI-assisted tools were used to support the debugging of code and the generation and refinement of some figures produced from it. All AI-generated content and suggestions were reviewed, verified, and edited by the authors, who take full responsibility for the accuracy, integrity, and originality of all results presented in the manuscript.


\newpage
\begin{appendices}
\counterwithin{figure}{section}

\section{Proof of \texorpdfstring{\Cref{thm:laspop}}{Lg}: LAS of DFE of Epidemiological Model \texorpdfstring{\eqref{eq:population_level}}{Lg}}\label{app:LAS}

\begin{proof}
    The following proof of local asymptotic stability of the disease-free equilibrium $\mathcal{E}_0$, given by \eqref{eq:DFEpop}, of the epidemiological model \eqref{eq:population_level} is entails the linearization of the model \eqref{eq:population_level} around the disease-free equilibrium \cite{martcheva2015introduction}.

Let
\begin{equation}\label{app:LAS1}
\begin{aligned}
S(t) & = S_0 + u(t), \quad \quad E(t) = v(t),
\quad\quad 
i(\tau,t) = w(\tau,t),\\
A(t) &= x(t), 
\quad \quad \quad \quad
R(t) = y(t), \text{ and } N(t) = S_0 + n(t).
\end{aligned}
\end{equation}

We note that (using approximation $\frac{1}{1+\epsilon} \approx 1-\epsilon$ and dropping higher order non-linear terms) \cite{martcheva2015introduction}:
$$\frac{S}{N} = \frac{S_0 + u(t)}{S_0+n(t)} = \frac{S_0+u(t)}{S_0}\left(\frac{1}{1+\dfrac{n(t)}{S_0}}\right) \approx \frac{S_0+u(t)}{S_0}\left(1-\frac{n(t)}{S_0}\right) \approx 1 + \frac{u(t)}{S_0}-\frac{n(t)}{S_0}.
$$    

We linearize the epidemiological model \eqref{eq:population_level} using \eqref{app:LAS1} to get (note that $S_0 = \frac{\Lambda}{\mu} \implies \Lambda = \mu S_0$:

\begin{equation}\label{app:LAS2}
\begin{aligned}
u'(t)  &= - \beta \left[\int_0^\infty \eta(\tau) w(\tau,t) \,d\tau + x(t)\right] - \mu u(t),\\
v'(t) &= \beta \left[\int_0^\infty \eta(\tau) w(\tau,t) d\tau + x(t)\right] - (\sigma+\mu)v(t),\\
w_\tau + w_t &= - (\gamma_i(\tau)+\mu) w(\tau,t), \quad \quad \quad \quad \quad w(0,t) = \rho \sigma v(t),\\
x'(t) &= (1-\rho)\sigma v(t) - (\gamma_A+\mu) x(t),\\
y'(t) &= \int_0^\infty \gamma_i(\tau) w(\tau,t) \,d\tau +\gamma_A x(t) - \mu y(t).
\end{aligned}
\end{equation}

Since \eqref{app:LAS2} is a linear model, we look for solution of the form:

\begin{equation}\label{app:LAS3}
    \begin{aligned}
u(t) = u_0 e^{\lambda t}, \quad v(t) = v_0 e^{\lambda t}, \quad w(\tau,t) = w(\tau)e^{\lambda t} , \quad
x(t) = x_0e^{\lambda t}, \quad \text{ and } \quad y(t) = y_0 e^{\lambda t}.
    \end{aligned}
\end{equation} 

Plugging \eqref{app:LAS3} into the system \eqref{app:LAS2}  gives the following eigenvalue problem: 

\begin{subequations} \label{eq:system}
\begin{align}
\lambda u_0  &= - \beta \left[\int_0^\infty \eta(\tau) w(\tau) \,d\tau + x_0\right] - \mu u_0 \label{eq:4i} \\
\lambda v_0 &= {\beta} \left[\int_0^\infty \eta(\tau) w(\tau) d\tau + x_0\right] - (\sigma+\mu)v_0 \label{eq:4ii} \\
w_\tau(\tau) + \lambda w(\tau) &= - (\gamma_i(\tau)+\mu) w(\tau), \quad w(0) = \rho \sigma v_0 \label{eq:4iii} \\
\lambda x_0 &= (1-\rho)\sigma v_0 - (\gamma_A+\mu) x_0 \label{eq:4iv} \\
\lambda y_0 &= \int_0^\infty \gamma_i(\tau) w(\tau) \,d\tau +\gamma_A x_0 - \mu y_0 \label{eq:1v}
\end{align}
\end{subequations}

Solving Equation \eqref{eq:4iii} (along with its initial condition) yields:
\begin{equation}\label{app:LAS5}
    w(\tau)=w(0)\pi(\tau)e^{-\lambda \tau}=\rho\sigma v_0\pi(\tau)e^{-\lambda \tau}
\end{equation}

where 
\[
\pi(\tau)=\exp{\left(-\int_0^\tau(\gamma_i({{s}})+\mu) d{{s}}\right)}
\]
is the probability of a symptomatic infectious individual remaining infectious $\tau$ time units after infection. 

Likewise, solving Equation \eqref{eq:4iv} yields:
\begin{equation}\label{app:LAS6}
x_0 = \frac{(1-\rho)\sigma}{\lambda+\gamma_A+\mu} v_0. 
\end{equation}

Plugging expressions for $w(\tau)$ and $x_0$ from equations \eqref{app:LAS5} and \eqref{app:LAS6}, respectively into Equation \eqref{eq:4ii}, leads to the characteristic equation $\mathcal{G}(\lambda) =1$, where 


\begin{equation}\label{eq:Glambda}
G(\lambda) = \frac{{\beta}}{\lambda+\sigma+\mu}\rho\sigma\int_0^\infty\eta(\tau)\pi(\tau)e^{-\lambda\tau }d\tau+\left(\frac{{\beta}}{\lambda+\sigma+\mu}\right)\left(\frac{(1-\rho)\sigma}{\lambda +\gamma_A+\mu}\right).
\end{equation}
{The DFE that is given in \ref{eq:DFEpop} is locally asymptotically stable if all eigenvalues (i.e., $\lambda$ solve $G(\lambda)=1$) have negative real parts, and the DFE is unstabile if there exists at least one eigenvalue with positive real part. }

Note that at $\lambda=0$ we  get the following generation of the basic reproduction number:
\begin{equation}
    G(0)=\frac{{\beta}}{\sigma+\mu}\rho\sigma\int_0^\infty\eta(\tau)\pi(\tau)d\tau + \left(\frac{{\beta}}{\sigma+\mu}\right)\left(\frac{(1-\rho)\sigma}{\gamma_A+\mu}\right)= \mathcal{R}_{C,i}+\mathcal{R}_{C,A} = \mathcal{R}_C.
    \label{eqn: R0}
\end{equation}

Let 
\begin{equation*}
     G_i(\lambda) = \frac{{\beta}}{\lambda+\sigma+\mu}\rho\sigma\int_0^\infty\eta(\tau)\pi(\tau)e^{-\lambda\tau }d\tau
    \quad \quad
    \text{and}
    \quad \quad 
     G_A(\lambda) = \left(\frac{{\beta}}{\lambda+\sigma+\mu}\right)\left(\frac{(1-\rho)\sigma}{\lambda +\gamma_A+\mu}\right) 
\end{equation*}
We can rewrite $G_i(\lambda$)
and $G_A(\lambda$) as a function of $\mathcal{R}_{C,i}$ and $\mathcal{R}_{C,A}$ (provided that there is some positive interval for which $\eta(\tau)$ is strictly positive), respectively, as:
\begin{equation*}
\begin{aligned}
G_i(\lambda) &= \mathcal{R}_{C,i} \left(\frac{\sigma+\mu}{\lambda+\sigma+\mu}\right) \int_0^\infty\eta(\tau)\pi(\tau)e^{-\lambda\tau }d\tau \frac{1}{\int_0^\infty\eta(\tau)\pi(\tau)d\tau} \quad \text{and} \\
     G_A(\lambda) &= \mathcal{R}_{C,A} \left(\frac{\sigma+\mu}{\lambda+\sigma+\mu}\right)  \left(\frac{\gamma_A+\mu}{\lambda+\gamma_A+\mu}\right) .
\end{aligned}
\end{equation*}
{Let $\lambda\in\mathbb{R}$. Since $\eta(\tau)$ and $\pi(\tau)$ are strictly positive, the} 
\begin{align*}
\frac{d G(\lambda) }{d\lambda}=-\frac{\beta}{(\lambda+\mu+\sigma)^2}&\left[\rho\sigma\int_0^{\infty}\eta(\tau)\pi(\tau)e^{-\lambda t}d\tau+\frac{(1-\rho)\sigma}{\lambda+\gamma_A+\mu}\right]\\
&-\left(\frac{\beta}{\lambda+\sigma+\mu}\right)\left[\rho\sigma\int_0^{\infty}\eta(\tau)\pi(\tau)\tau e^{-\lambda\tau}d\tau+\frac{(1-\rho)\sigma}{(\lambda+\gamma_A+\mu)^2}\right]
\end{align*}
is negative whenever $\lambda>0$. Thus $G$ is a decreasing function of $\lambda$ whenever $\lambda$ is positive, and 
\[
\lim\limits_{\lambda\to\infty} G(\lambda)=0.
\]
When $\mathcal{R}_C=G(0)>1$, there exists a unique positive real solution to the equation $G(\lambda)=1$, by the Intermediate Value Theorem.
Next we consider the case where $\mathcal{R}_0=G(0)<1$. Let $\lambda=\alpha+bi$ be an eigenvalue with nonnegative real part ($\alpha\geq 0$). Then,
\begin{align}
1&=|G(\alpha+bi)|\nonumber\\
&\leq |G_i(\alpha+bi)|+|G_A(\alpha+bi)|\nonumber\\
&\leq \mathcal R_{C,i}\frac{\sigma+\mu}{|\lambda+\sigma+\mu|}\int_0^{\infty}\eta(\tau)\pi(\tau)e^{-\alpha \tau}|e^{-b\tau i}|d\tau\left(\frac{1}{\int_0^{\infty}\eta(\tau)\pi(\tau)d\tau}\right)+\mathcal R_{C,A}\frac{\sigma+\mu}{|\lambda+\sigma+\mu|}\frac{\gamma_A+\mu}{|\lambda+\gamma_A+\mu|}\nonumber\\
&\leq \mathcal R_{C,i}\frac{\sigma+\mu}{|\lambda+\sigma+\mu|}\left(\frac{\int_0^{\infty}\eta(\tau)\pi(\tau)e^{-\alpha \tau}d\tau}{\int_0^{\infty}\eta(\tau)\pi(\tau)d\tau}\right)+\mathcal R_{C,A}\frac{\sigma+\mu}{|\lambda+\sigma+\mu|}\frac{\gamma_A+\mu}{|\lambda+\gamma_A+\mu|}.\label{eqn: lasineq}
\end{align}
Since $\alpha\geq0,\;\eta(\tau)\geq 0$, and $\pi(\tau)\geq 0$ for all $\tau\geq0$, $\eta(\tau)\pi(\tau)e^{-\alpha \tau}\leq \eta(\tau)\pi(\tau)$ for all $\tau\geq0$ which implies
\[\frac{\int_0^{\infty}\eta(\tau)\pi(\tau)e^{-\alpha \tau}d\tau}{\int_0^{\infty}\eta(\tau)\pi(\tau)d\tau}\leq 1.\] In addition nonnegativity of the real part of $\lambda$ ($\alpha\geq0$) we have the following hold any positive real number $x$:
\[
\frac{x}{|\lambda+x|}\leq\frac{x}{\sqrt{(\alpha+x)^2+b^2}}\leq \frac{x}{(\alpha+x)}\leq 1.
\]
Thus, inequality (\ref{eqn: lasineq}) simplifies to 
\[
1=|G(\alpha+bi)|\leq \mathcal{R}_{C}<1,
\]
which is a contradiction. Therefore, all eigenvalues have negative real parts when $\mathcal{R}_0<1$. This completes the proof.

\end{proof}

\newpage
\section{Proof of \texorpdfstring{\Cref{thm:endemicUNIQUE}}{Lg}: Unique Endemic Equilibrium of Epidemiological Model \texorpdfstring{\eqref{eq:population_level}}{Lg}}\label{app:endemicunique}
\begin{theorem}
    If $\mathcal{R}_C>1$, then system (\ref{eq:population_level}) has a unique endemic equilibrium. 
    \begin{align*}
        S^*&=\frac{N^*}{\mathcal{R}_C}= \frac{\Lambda}{\mu}\cdot\frac{1}{\mathcal{R}_C}\\
        E^*&=\frac{\Lambda}{\sigma+\mu}\left(1-\frac{1}{\mathcal{R}_C}\right)\\
        i^*(\tau)&=i^*(0)\pi(\tau) = \rho\sigma\pi(\tau)E^*\\
        A^*&=\frac{(1-\rho)\sigma}{\gamma_A+\mu}E^*\\
        R^*&=\frac{1}{\mu}\left[\rho\sigma\int_0^{\infty}\gamma_i(\tau)\pi(\tau)d\tau+\frac{\gamma_A(1-\rho)\sigma}{\gamma_A+\mu}\right]E^*
    \end{align*}
    where $\pi(\tau)=\exp\left(-\int_0^{{\tau}}(\gamma_i({{\widehat{\tau}}})+\mu)d{\widehat{\tau}}\right)$
    
\end{theorem}
\begin{proof}
\textbf{Observation}: The total population $N(t)$ described in (\ref{eq:population_level}) is asymptotically constant. 
When integrating the PDE equation relating to the symptomatic infectious compartment in (\ref{eq:population_level}) with respect to $\tau$, we obtain 
\begin{align*}
    i(\tau,t)\rvert_0^{\infty}+\frac{dI}{dt}&=-\int_0^{\infty}\gamma_i(\tau)i(\tau,t)d\tau - \mu I(t).
\end{align*}
If we assume $\lim_{\tau\to\infty}i(\tau,t)=0$ and apply condition $i(0,t)=\rho\sigma E(t)$ from (\ref{eq:population_level}) we obtain
\begin{equation*}
    \frac{dI}{dt}=\rho\sigma E(t)-\int_0^\infty\gamma_i(\tau)d\tau-\mu I(t)
\end{equation*}
Adding the equation above to the other equations given in (\ref{eq:population_level}) yields
\begin{equation}
    \frac{dN}{dt}= \Lambda - \mu N
\end{equation}
Thus in looking at any time independent solution to system (\ref{eq:population_level}), we have that $N^*=\frac{\Lambda}{\mu}$

We consider a time independent solution to system (\ref{eq:population_level}) for which $E^*\neq0$
\begin{subequations}\label{eq:steadystateeq}
\begin{align}
0 &= \Lambda -\beta \left( \displaystyle\int_0^\infty \eta(Y(\tau))i^*(\tau) \, d\tau + A^* \right) \frac{S^*}{N^*} - \mu S^*, \label{eqn: S_RHS}\\
0 &= \beta \left( \displaystyle\int_0^\infty \eta(Y(\tau))i^*(\tau) \, d\tau + A^* \right) \frac{S^*}{N^*} - (\sigma+\mu) E^*, \label{eqn: E_RHS}\\
\frac{d i^*(\tau)}{d \tau} &= -(\gamma_i(\tau)+\mu) i^*(\tau), \quad \quad  \quad  i^*(0) = \rho \sigma E^*, \label{eqn: I_RHS}\\
0 &= (1-\rho) \sigma E^* - (\gamma_A+\mu) A^*, \label{eqn: A_RHS}\\
0 &= \int_0^\infty \gamma_i(\tau) i^*(\tau) \, d\tau + \gamma_A A^* - \mu R^*\label{eqn: R_RHS}.
\end{align}
\end{subequations}
From (\ref{eqn: I_RHS}) where have 
\begin{equation}
    i^*(\tau)=i^*(0)\pi(\tau) = \rho\sigma\pi(\tau)E^*\label{eqn: istar}
\end{equation}
where $\pi(\tau)=e^{-\int_0^{\tau}(\gamma_i({\widehat{\tau}})+\mu) d{\widehat{\tau}}}$. 
From (\ref{eqn: A_RHS}) we have 
\begin{equation}
    A^*= \frac{(1-\rho)\sigma}{\gamma_A+\mu}E^*. \label{eqn: Astar1}
\end{equation}
Substituting (\ref{eqn: istar}) and (\ref{eqn: Astar1}) into (\ref{eqn: R_RHS}) yields 
\begin{equation}
    R^*= \frac{1}{\mu}\left[\rho\sigma\int_0^{\infty}\gamma_i(\tau)\pi(\tau)d\tau+\frac{\gamma_A(1-\rho)\sigma}{\gamma_A+\mu}\right]E^*. \label{eqn: Rstar}
\end{equation}
Substituting (\ref{eqn: istar}) and (\ref{eqn: Astar1}) into equation (\ref{eqn: E_RHS}) gives 
\begin{equation*}
    0=\left(\beta\rho\sigma\int_0^{\infty}{\eta(Y(\tau))}\pi(\tau)d\tau+\beta\frac{(1-\rho)\sigma}{\gamma_A+\mu}\right)\frac{S^*}{N^*}E^*-(\sigma+\mu)E^*.
\end{equation*}
By dividing the above equation by $(\sigma+\mu)E^*$ and rearranging terms, we obtain an equation in terms of $\mathcal{R}_C$ (\ref{eqn: R0}) 
\begin{equation}\label{eqn: Sstar}
    S^*=\frac{N^*}{\mathcal{R}_C}
\end{equation}
where $N^*=\frac{\Lambda}{\mu}$.
Applying equations (\ref{eqn: istar}), (\ref{eqn: Astar1}), and (\ref{eqn: Sstar}) into equation (\ref{eqn: S_RHS}) yields 
\begin{equation}
   {0=}\Lambda-(\sigma+\mu)E^*-\frac{\mu N^*}{\mathcal{R}_C}
\end{equation}
Rearranging terms and using $N^*=\frac{\Lambda}{\mu}$ gives
\begin{equation}\label{eqn: Estar}
E^*= \frac{\Lambda}{\sigma+\mu}\left(1-\frac{1}{\mathcal{R}_C}\right)
\end{equation}
$E^*>0$ if and only if $\mathcal{R}_C>1$.
\end{proof}

\newpage
\section{Numerical Scheme}\label{numerical_scheme}
The numerical scheme presented is a mixed implicit scheme. 
Let $T$ and $\mathcal{T}$ be numbers relating to the final epidemic time and final time since infection time (respectively) that we use in generating numerical solutions for the multiscale model \model. We partition the time since infection time interval $[0,\mathcal{T}]$ using $\hat{K}+1$ equally spaced nodes $0=\tau_0\leq\tau_1\leq\dots\leq\tau_{\hat{K}}=\mathcal{T}$. In particular $\tau_k=k\Delta \tau$  (for all $0\leq k\leq \hat{K}$) with mesh size $\Delta \tau=\frac{\mathcal{T}}{\hat{K}}$. Similarly we partition the epidemic time interval $[0,T]$, using $\hat{N}+1$ equally spaced nodes $\{t_n\}_{n=0}^{\hat{N}}$ ($t_n= n\Delta t$ for all $0\leq n\leq \hat{N}$) with mesh size $\Delta t=\frac{T}{\hat{N}}$. We choose partitions for each time interval to where the mesh sizes are the same, $\Delta t=\Delta \tau$ (and we use $\Delta t$ exclusively throughout the numerical scheme). Additionally, partitions are chosen to  where numbers $\hat{K}=\frac{\mathcal{T}}{\Delta t}$ and $\hat{N}=\frac{T}{\Delta t}$ are integers.\\
 We use the following to notate our approximations: 
$S^n\approx S(t_n)$, $E^n\approx E(t_n)$, $A^n\approx A(t_n)$, $R^n\approx R(t_n)$, $i_k^n\approx i(\tau_k,t_n)$, and $N^n\approx N(t_n)\approx S^n+E^n+\sum\limits_{k=1}^{\hat{K}} i_k^n\Delta t+A^n+R^n$. We initialize all variables with the initial conditions that give the values at time level zero
\[S^0=S_0,\quad E^0=E_0,\quad A^0=A_0, \quad R^0=R_0, \quad \text{and } i_k^0=i_0(\tau_k) \; (\text{for }k=0,\dots,\hat{K})
\]
Additionally, let $\eta_{k}=\eta_i(\tau_k)$ and $\gamma_{i,k}=\gamma_i(\tau_k)$ [Note that within our simulations $\gamma_i(\tau)$ is set as some constant function, but the scheme presented here assumes that $\gamma_i(\tau)$ is dependent on the within-host dynamics.]. 
We solve the within-host model using \texttt{MATLAB}'s ODE solver \texttt{ode45} and store the within host solutions at the specified mesh points $\{\tau_k\}_{k=0}^{\hat{K}}$. These values are then used directly for computing $\eta_{k}$ and $\gamma_{i,k}$ for $k=0,\dots, \hat{K}$.  
We summarize the numerical method below. 
\begin{align*}
N^{n+1} &= \frac{N^n+\Lambda \Delta t}{1+\mu\Delta t } \\
S^{n+1} &= \frac{S^n+\Lambda\Delta t }{1+\mu\Delta t +\frac{\beta}{N^{n+1}}\left(\frac{\Delta t}{2}\left[\eta_0i^n_0+2\left(\sum\limits_{k=0}^{\hat{K}-1} \eta_ki_k^n\right)+\eta_{\hat{K}}i^n_{\hat{K}}\right]+A^n\right)\Delta t}\\
E^{n+1}&= \frac{E^n+\frac{\beta S^{n+1}}{N^{n+1}}\left(\frac{\Delta t}{2}\left[\eta_0i^n_0+2\left(\sum\limits_{k=0}^{\hat{K}-1} \eta_ki_k^n\right)+\eta_{\hat{K}}i^n_{\hat{K}}\right]+A^n\right)\Delta t}{1+(\sigma+\mu)\Delta t}\\
i^{n+1}_0&=\rho\sigma E^{n+1}\\
i^{n+1}_{k+1}&=\frac{i^n_k}{1+(\gamma_{i,k}+\mu)\Delta t}\\
A^{n+1}&=\frac{A^n+(1-\rho)\sigma E^{n+1}\Delta t}{1+(\gamma_A+\mu)\Delta t}\\
R^{n+1}&= \frac{R^n+\left(\frac{\Delta t}{2}\left[\gamma_{i,0}i_{0}^{n+1}+2\left(\sum\limits_{k=1}^{\hat{K}-1}\gamma_{i,k}i_{k}^{n+1}\right)+\gamma_{i,\hat{k}}i^{n+1}_{\hat{K}}\right]+\gamma_A A^{n+1}\right)\Delta t}{1+\mu\Delta t}
\end{align*}
Note that for approximating $\int_0^{\infty}\eta(\tau)i(\tau,t)d\tau$ and $\int_0^{\infty}\gamma_i(\tau)i(\tau,t)d\tau$, we apply trapezoidal rule (but with upper integral limit being replaced by a sufficiently large number  $\mathcal{T}$). 



\newpage
\section{Alternative Linking Functions}\label{app:linking}

This study primarily focuses on exponential form of linking function \eqref{eq:exponential}. In this section, we present other two commonly used linking functions (namely Hill \cite{fraser2007variation,atkins2023immuno,saldana2025multiscale, martcheva2015introduction} and linear functions \cite{gilchrist2002modeling,saldana2025multiscale, martcheva2015introduction}) alongside a simple threshold based on/off linking function.

\textbf{Generalized Hill function.} A generalized Hill function that links the two scales can be express as:
\begin{equation}\label{eq:hill}
    \eta(Y) = 1 - a\frac{Y^k}{b^k + Y^k},
\end{equation}
where $a \in (0,1]$, $b > 0$, and $k$ is a positive integer (we investigate $k = 1, 2$). The parameter $(1-a)$ is an explicit transmission floor, representing residual transmission that persists even under maximal symptom severity (capturing, for instance, household exposure to a bedridden individual). The parameter $a$ thus has a direct interpretation as the maximum achievable behavioral modification, ranging from mild social distancing ($a$ small) to strict quarantine ($a$ close to 1). The half-saturation constant $b$ gives the symptom score at which half of the maximum behavioral reduction is achieved. For simplcity, we assume $b=1$ for simulation purposes. 

\paragraph{Linear function.} A linear linking function that connects the two scales can be expressed as:
\begin{equation}
\eta(Y) = \max(1-aY,\ 0),
\label{eq:linear_link}
\end{equation}
where $a>0$ controls the rate at which increasing symptom scores reduce contact. 

\paragraph{Threshold based on/off function.} A simple threshold-based linking function, representing a naive isolation mechanism, can be expressed as:
\begin{equation}
\eta(Y) =
\begin{cases}
1, & Y(\tau) < Y_{\text{threshold}}, \\
0, & Y(\tau) \geq Y_{\text{threshold}},
\end{cases}
\label{eq:threshold_link}
\end{equation}
where $Y_{\text{threshold}}>0$ is the symptom score at which an individual fully isolates. 

\begin{figure}[htbp]
    \centering
    \includegraphics[width=\textwidth]{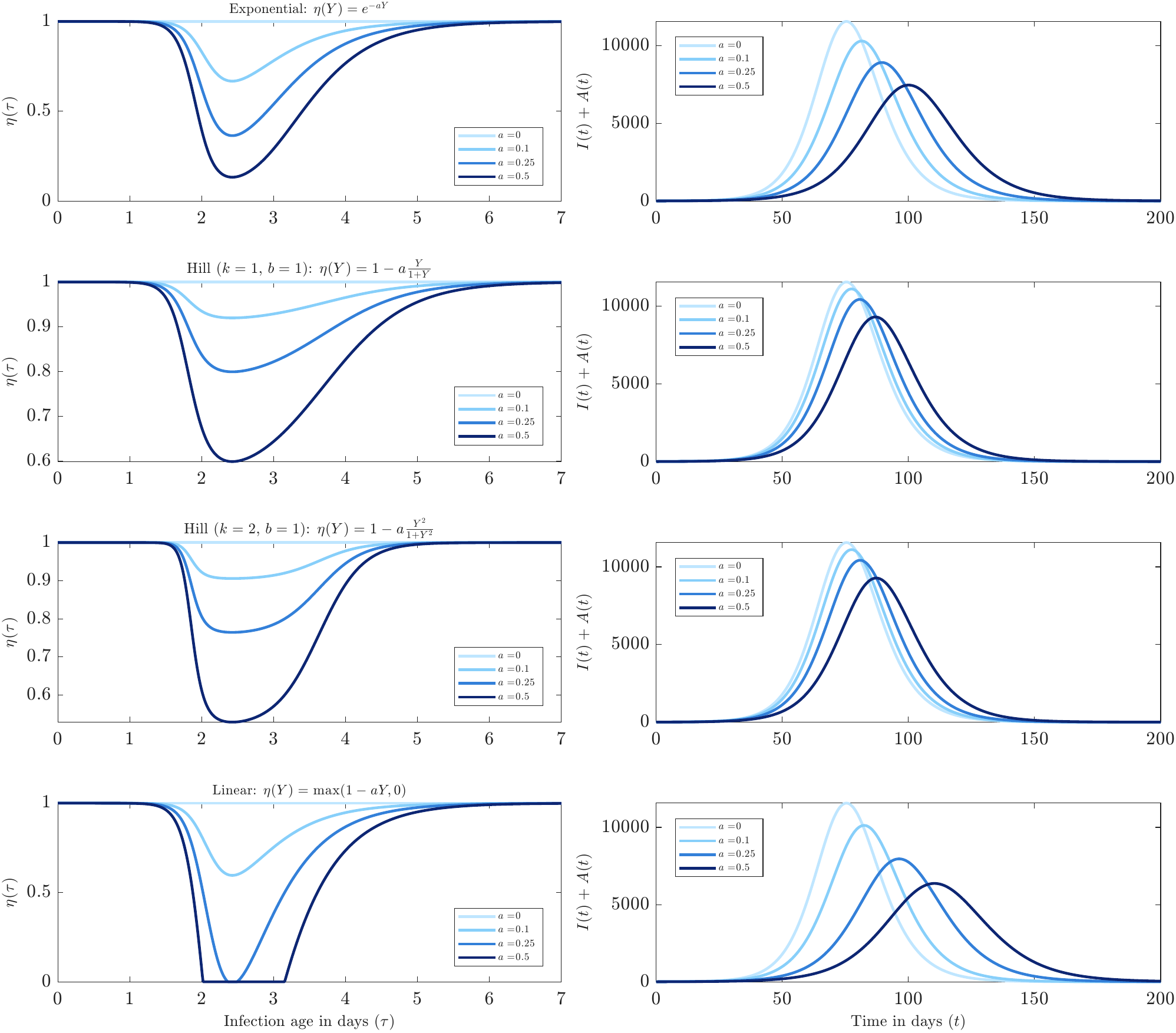}
\caption{Effect of the different choice of multi-scale linking function on contact reduction and epidemic dynamics. Each row corresponds to a different functional form for the linking function $\eta(\tau)$,: (top to bottom) exponential; Hill-type with $k=1$, $b=1$; rational with $k=2$, $b=1$; and linear. In each row, curves show the behavior-change parameter $a\in\{0,0.1,0.25,0.5\}$ (light to dark blue), with $a=0$ corresponding to no behavioral response ($\eta\equiv1$). Left column: the resulting contact-reduction function $\eta(\tau)$ as a function of infection age $\tau$. Right column: the corresponding total infectious population $I(t)+A(t)$.}
\label{fig:linking_functions_comparison}
\end{figure}

\newpage

\begin{figure}[htbp]
    \centering
\includegraphics[width=1\textwidth]{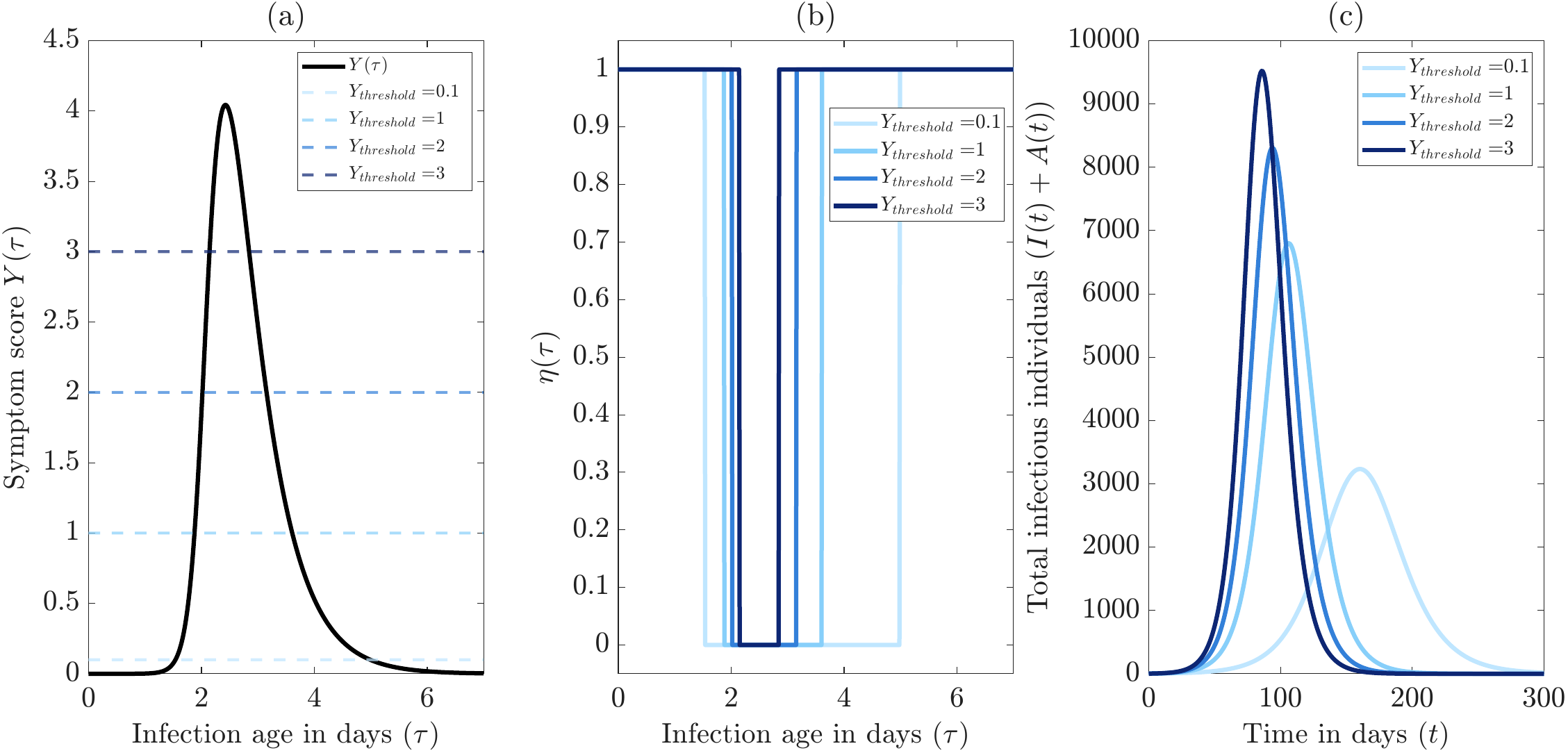}
\caption{Effect of the symptom-score based threshold linking function on population-level disease dynamics. (a) Within-host symptom score $Y(\tau)$ as a function of infection age $\tau$, with horizontal dashed lines marking candidate thresholds $Y_{threshold} = 0.1, 1, 2, 3$. (b) Corresponding behavioral response function $\eta(\tau)$: individuals avoid contacts only while $Y(\tau)$ exceeds the threshold (i.e., $\eta$ drops to 0), so lower thresholds trigger earlier onset and longer duration of behavior change. (c) Resulting population-level epidemic curves, total infectious individuals $I(t) + A(t)$, under each threshold. Lower thresholds (earlier, more prolonged symptom-driven behavior change) produce smaller, delayed, and more spread-out epidemic peaks, while higher thresholds, which correspond to little to no behavioral response, yield larger, earlier, and sharper peaks.}
\label{fig:isolation_linking}
\end{figure}


\newpage
\section{SEIR Model with Behavior Change}\label{SEIRb}

We adapt one of the models presented in Pant et al. \cite{pant2026paradox} to demonstrate how models that account for population-level behavior change due to disease-information at a population level (here, number of infectious individuals) can lead to heightened peak suppression, while slower and elongated epidemic decline with an increase in behavior change (see Figure \ref{fig:seir-behavior}).  The adapted model is given by \cite{pant2026paradox}:

\begin{equation}
\left\{
\begin{aligned}
\frac{dS}{dt} &= \Lambda -\beta_0 e^{-aI}\frac{SI}{N} -\mu S, \\
\frac{dE}{dt} &= \beta_0 e^{-aI}\frac{SI}{N} - (\sigma+\mu) E, \\
\frac{dI}{dt} &= \sigma E - (\gamma+\mu) I, \\
\frac{dR}{dt} &= \gamma I - \mu R,
\end{aligned}
\right.
\label{eq:seir-behavior}
\end{equation}
where parameter $a$ is the population-level behavior change parameter that can be viewed as a sensitivity to population-level knowledge of the number of infectious individuals. This interpretation of parameter \(a\) in the adapted model contrasts with its role in the multiscale model \model. In the latter, although parameter \(a\) still governs population-level behavior change, it operates exclusively on symptomatic individuals and is driven solely by the severity of the symptoms they experience.

\begin{figure}[htbp]
    \centering
    \includegraphics[width=\textwidth]{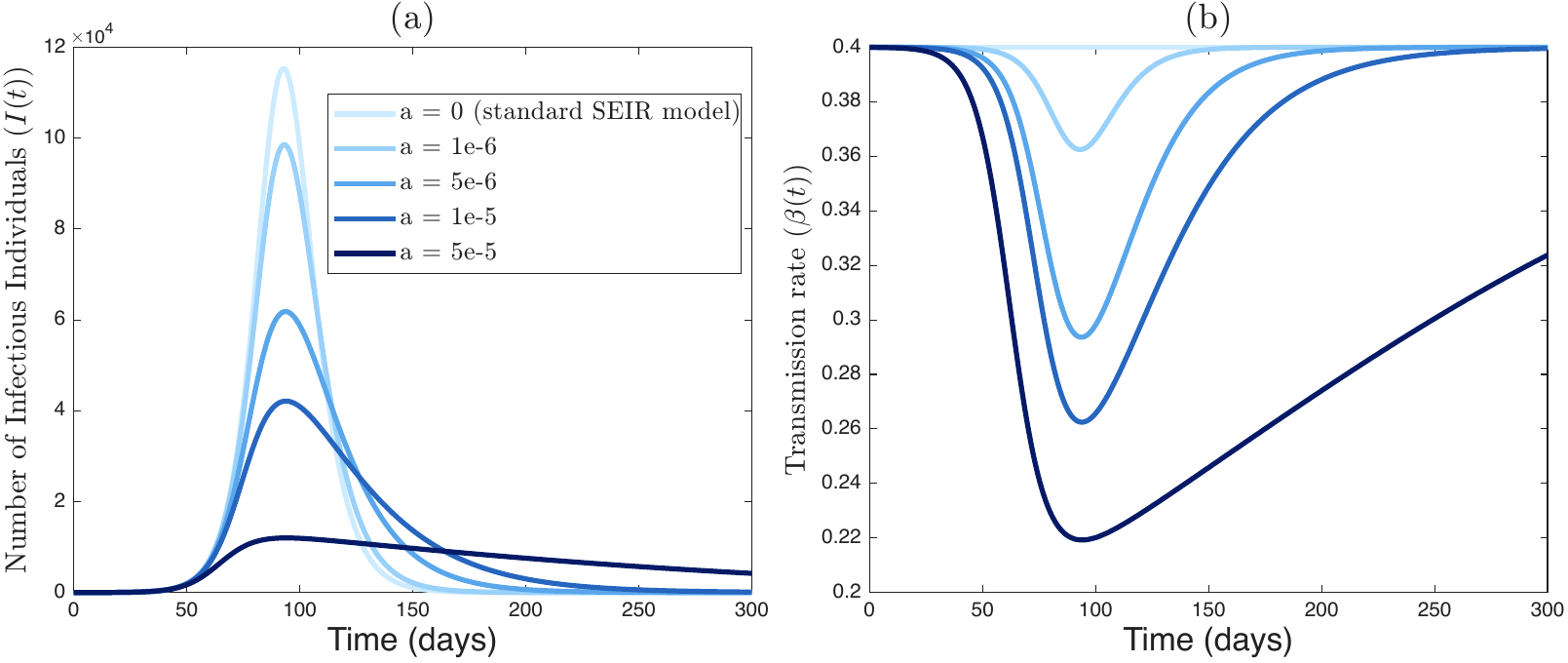}
    \caption{Output of Model~\eqref{eq:seir-behavior} showcasing the impact of increasing behavior change parameter on (a) the number of infectious individuals, and (b) the corresponding transmission rate over the same period. Setting $a=0$ in \eqref{eq:seir-behavior}  recovers the standard SEIR model with constant transmission rate $\beta_0$; increasing parameter $a$ produces results in a flattened, delayed epidemic peak and a slower, more prolonged decline. In this figure, the behavior change parameter is varied while the rest of the parameters are kept fixed at the same values as in \Cref{tab:param_value}, specifically, $\Lambda = 3.5215~\text{day}^{-1}$, $\mu = 1/(77.8\times 365)~\text{day}^{-1}$, $\sigma = 1/1.6~\text{day}^{-1}$, $\gamma = 1/5~\text{day}^{-1}$, $\beta_0 = 0.4~\text{day}^{-1}$, $N = 10^6$, $S(0) = 999,995$, $E(0)=5$, and $I(0)=R(0)=0$. This figure is inspired by a Figure 1 presented in Pant et al. \cite{pant2026paradox}.}
    \label{fig:seir-behavior}
\end{figure}

\setcounter{figure}{0} 
\end{appendices}
\newpage
\bibliographystyle{ieeetr}
\bibliography{references.bib}

@article{canini2011population,
  title={Population modeling of influenza A/H1N1 virus kinetics and symptom dynamics},
  author={Canini, Laetitia and Carrat, Fabrice},
  journal={Journal of virology},
  volume={85},
  number={6},
  pages={2764--2770},
  year={2011},
  publisher={Am Soc Microbiol}
}

@article{pant2026paradox,
  title={The paradox of neglecting changes in behavior: How standard epidemic models misestimate both transmissibility and final epidemic size},
  author={Pant, Binod and Lalovic, Marko and Kiss, Istv{\'a}n Z and Santillana, Mauricio},
  journal={PLOS Computational Biology},
  volume={22},
  number={9},
  pages={e1014746},
  year={2026},
  publisher={Public Library of Science San Francisco, CA USA}
}

@article{pant2026uncovering,
  title={Uncovering identifiability of epidemiological models: basic reproduction number and complementary data streams},
  author={Pant, Binod and Saucedo, Omar and Pogudin, Gleb},
  journal={medRxiv},
  pages={2026--01},
  year={2026},
  publisher={Cold Spring Harbor Laboratory Press}
}

@article{diekmann1990definition,
  title={On the definition and the computation of the basic reproduction ratio R 0 in models for infectious diseases in heterogeneous populations},
  author={Diekmann, Odo and Heesterbeek, Johan Andre Peter and Metz, Johan Anton Jacob},
  journal={Journal of mathematical biology},
  volume={28},
  number={4},
  pages={365--382},
  year={1990},
  publisher={Springer}
}

@article{dobrovolny20asymp,
    title = {Modeling the role of asymptomatics in infection spread with application to {SARS-CoV}-2},
	author = {Hana M. Dobrovolny}, 
	journal = {Plos One},
	volume = {15},
	number = {8},
	month = {August 10},
	year = {2020},
	pages = {e0236976},
	doi = {10.1371/journal.pone.0236976}
}

@article{oveson2025modeling,
  title={Modeling the impact of hospitalization-induced behavioral changes on the spread of COVID-19 in New York City},
  author={Oveson, Alice and Girvan, Michelle and Gumel, Abba B},
  journal={Infectious Disease Modelling},
  volume={10},
  number={4},
  pages={1055--1092},
  year={2025},
  publisher={Elsevier}
}

@article{lejeune2024mathematical,
  title={Mathematical analysis of simple behavioral epidemic models},
  author={LeJeune, Leah and Ghaffarzadegan, Navid and Childs, Lauren M and Saucedo, Omar},
  journal={Mathematical Biosciences},
  volume={375},
  pages={109250},
  year={2024},
  publisher={Elsevier}
}

@article{saad2026interactions,
  title={Interactions between immuno-epidemiology and individual decision-making for nonpharmaceutical interventions},
  author={Saad-Roy, Chadi M and Abraham, Ninan and Hilbe, Christian and Mahmud, Ayesha S and Traulsen, Arne},
  journal={Trends in Microbiology},
  year={2026},
  publisher={Elsevier}
}

@article{berestycki2021plateaus,
  title={Plateaus, rebounds and the effects of individual behaviours in epidemics},
  author={Berestycki, Henri and Desjardins, Benoit and Heintz, Bruno and Oury, Jean-Marc},
  journal={Scientific reports},
  volume={11},
  number={1},
  pages={18339},
  year={2021},
  publisher={Nature Publishing Group UK London}
}

@article{weitz2020awareness,
  title={Awareness-driven behavior changes can shift the shape of epidemics away from peaks and toward plateaus, shoulders, and oscillations},
  author={Weitz, Joshua S and Park, Sang Woo and Eksin, Ceyhun and Dushoff, Jonathan},
  journal={Proceedings of the National Academy of Sciences},
  volume={117},
  number={51},
  pages={32764--32771},
  year={2020},
  publisher={National Academy of Sciences}
}

@article{cori2012estimating,
  title={Estimating influenza latency and infectious period durations using viral excretion data},
  author={Cori, Anne and Valleron, Alain-Jacques and Carrat, Fabrice and Tomba, G Scalia and Thomas, Gilles and Bo{\"e}lle, Pierre-Yves},
  journal={Epidemics},
  volume={4},
  number={3},
  pages={132--138},
  year={2012},
  publisher={Elsevier}
}

@article{urmi2024characterizing,
  title={Characterizing Population-level Changes in Human Behavior during the COVID-19 Pandemic in the United States},
  author={Urmi, Tamanna and Pant, Binod and Dewey, George and Quintana-Math{\'e}, Alexi and Lang, Iris and Druckman, James N and Ognyanova, Katherine and Baum, Matthew and Perlis, Roy H and Riedl, Christoph and others},
  journal={medRxiv},
  pages={2024--12},
  year={2024},
  publisher={Cold Spring Harbor Laboratory Press}
}

@article{maniscalco2024adaptive,
  title={Adaptive behavior in response to the 2022 mpox epidemic in the Paris region},
  author={Maniscalco, Davide and Robineau, Olivier and Bo{\"e}lle, Pierre-Yves and Mazzoli, Mattia and Barret, Anne-Sophie and Chazelle, Emilie and Mailles, Alexandra and No{\"e}l, Harold and Tarantola, Arnaud and Velter, Annie and others},
  journal={medRxiv},
  pages={2024--10},
  year={2024},
  publisher={Cold Spring Harbor Laboratory Press}
}

@article{kiss2010impact,
  title={The impact of information transmission on epidemic outbreaks},
  author={Kiss, Istvan Z and Cassell, Jackie and Recker, Mario and Simon, P{\'e}ter L},
  journal={Mathematical biosciences},
  volume={225},
  number={1},
  pages={1--10},
  year={2010},
  publisher={Elsevier}
}

@article{gilchrist2002modeling,
  title={Modeling host--parasite coevolution: a nested approach based on mechanistic models},
  author={Gilchrist, Michael A and Sasaki, Akira},
  journal={Journal of Theoretical Biology},
  volume={218},
  number={3},
  pages={289--308},
  year={2002},
  publisher={Elsevier}
}

@article{ma2021global,
  title={Global percentage of asymptomatic SARS-CoV-2 infections among the tested population and individuals with confirmed COVID-19 diagnosis: a systematic review and meta-analysis},
  author={Ma, Qiuyue and Liu, Jue and Liu, Qiao and Kang, Liangyu and Liu, Runqing and Jing, Wenzhan and Wu, Yu and Liu, Min},
  journal={JAMA network open},
  volume={4},
  number={12},
  pages={e2137257--e2137257},
  year={2021},
  publisher={American Medical Association}
}

@article{moghadas2020implications,
  title={The implications of silent transmission for the control of COVID-19 outbreaks},
  author={Moghadas, Seyed M and Fitzpatrick, Meagan C and Sah, Pratha and Pandey, Abhishek and Shoukat, Affan and Singer, Burton H and Galvani, Alison P},
  journal={Proceedings of the National Academy of Sciences},
  volume={117},
  number={30},
  pages={17513--17515},
  year={2020},
  publisher={National Academy of Sciences}
}

@article{prochazka2024temporary,
  title={Temporary adaptations to sexual behaviour during the mpox outbreak in 23 countries in Europe and the Americas: findings from a retrospective cross-sectional online survey},
  author={Prochazka, Mateo and Vinti, Pietro and Hoxha, Ana and Seale, Andy and Mozalevskis, Antons and Lewis, Rosamund and Sagastume, Ruben Mayorga and Scherzer, Martha and Dore, Leilia and Doherty, Meg},
  journal={The Lancet Infectious Diseases},
  volume={24},
  number={12},
  pages={1309--1318},
  year={2024},
  publisher={Elsevier}
}

@article{mclean2017school,
  title={School absenteeism among school-aged children with medically attended acute viral respiratory illness during three influenza seasons, 2012-2013 through 2014-2015},
  author={McLean, Huong Q and Peterson, Siri H and King, Jennifer P and Meece, Jennifer K and Belongia, Edward A},
  journal={Influenza and other respiratory viruses},
  volume={11},
  number={3},
  pages={220--229},
  year={2017},
  publisher={Wiley Online Library}
}

@misc{cnn_1918_masks_2020,
  author = {{CNN}},
  title = {In the 1918 flu pandemic, not wearing a mask was illegal in some parts of America. What changed?},
  year = {2020},
  month = {April},
  day = {5},
  url = {https://www.cnn.com/2020/04/03/americas/flu-america-1918-masks-intl-hnk},
  urldate = {2025-08-29},
  note = {Accessed: August 29, 2025}
}

@book{martcheva2015introduction,
  title={An introduction to mathematical epidemiology},
  author={Martcheva, Maia},
  volume={61},
  year={2015},
  publisher={Springer}
}

@article{pant2024behavior,
  title={Mathematical assessment of the role of human behavior changes on SARS-CoV-2 transmission dynamics in the united states},
  author={Pant, Binod and Safdar, Salman and Santillana, Mauricio and Gumel, Abba B},
  journal={Bulletin of mathematical biology},
  volume={86},
  number={8},
  pages={92},
  year={2024},
  publisher={Springer}
}

@article{argyropoulos2020association,
  title={Association of initial viral load in severe acute respiratory syndrome coronavirus 2 (SARS-CoV-2) patients with outcome and symptoms},
  author={Argyropoulos, Kimon V and Serrano, Antonio and Hu, Jiyuan and Black, Margaret and Feng, Xiaojun and Shen, Guomiao and Call, Melissa and Kim, Min J and Lytle, Andrew and Belovarac, Brendan and others},
  journal={The American journal of pathology},
  volume={190},
  number={9},
  pages={1881--1887},
  year={2020},
  publisher={Elsevier}
}

@article{crystal2008airway,
  title={Airway epithelial cells: current concepts and challenges},
  author={Crystal, Ronald G and Randell, Scott H and Engelhardt, John F and Voynow, Judith and Sunday, Mary E},
  journal={Proceedings of the American Thoracic Society},
  volume={5},
  number={7},
  pages={772--777},
  year={2008},
  publisher={American Thoracic Society}
}

@article{challenger2022modelling,
  title={Modelling upper respiratory viral load dynamics of SARS-CoV-2},
  author={Challenger, Joseph D and Foo, Cher Y and Wu, Yue and Yan, Ada WC and Marjaneh, Mahdi Moradi and Liew, Felicity and Thwaites, Ryan S and Okell, Lucy C and Cunnington, Aubrey J},
  journal={BMC medicine},
  volume={20},
  number={1},
  pages={25},
  year={2022},
  publisher={Springer}
}

@article{he2020temporal,
  title={Temporal dynamics in viral shedding and transmissibility of COVID-19},
  author={He, Xi and Lau, Eric HY and Wu, Peng and Deng, Xilong and Wang, Jian and Hao, Xinxin and Lau, Yiu Chung and Wong, Jessica Y and Guan, Yujuan and Tan, Xinghua and others},
  journal={Nature medicine},
  volume={26},
  number={5},
  pages={672--675},
  year={2020},
  publisher={Nature Publishing Group US New York}
}

@article{fraser2007variation,
  title={Variation in HIV-1 set-point viral load: epidemiological analysis and an evolutionary hypothesis},
  author={Fraser, Christophe and Hollingsworth, T D{\'e}irdre and Chapman, Ruth and de Wolf, Frank and Hanage, William P},
  journal={Proceedings of the National Academy of Sciences},
  volume={104},
  number={44},
  pages={17441--17446},
  year={2007},
  publisher={National Academy of Sciences}
}

@article{atkins2023immuno,
  title={An Immuno-Epidemiological Model of Foot-and-Mouth Disease (FMD) in the African Buffalo Population with Carrier Transmission},
  author={Atkins, Summer and Gulbudak, Hayriye and Welker, J Shane and Smith, Houston},
  journal={bioRxiv},
  pages={2023--09},
  year={2023},
  publisher={Cold Spring Harbor Laboratory}
}

@article{saldana2025multiscale,
  title={Multiscale Modeling of Vector-Borne Diseases: The Role of Dose-Dependent Transmission},
  author={Salda{\~n}a, Fernando and Velasco-Hern{\'a}ndez, Jorge X and Ezanno, Pauline and Cecilia, H{\'e}l{\`e}ne},
  journal={bioRxiv},
  pages={2025--08},
  year={2025},
  publisher={Cold Spring Harbor Laboratory}
}

@article{gumel2021primer,
  title={A primer on using mathematics to understand COVID-19 dynamics: Modeling, analysis and simulations},
  author={Gumel, Abba B and Iboi, Enahoro A and Ngonghala, Calistus N and Elbasha, Elamin H},
  journal={Infectious Disease Modelling},
  volume={6},
  pages={148--168},
  year={2021},
  publisher={Elsevier}
}

@article{pant2024hetro,
  title={Mathematical assessment of the roles of age heterogeneity and vaccination on the dynamics and control of SARS-CoV-2},
  author={Pant, Binod and Gumel, Abba B},
  journal={Infectious Disease Modelling},
  volume={9},
  number={3},
  pages={828--874},
  year={2024},
  publisher={Elsevier}
}

@article{van2002reproduction,
  title={Reproduction numbers and sub-threshold endemic equilibria for compartmental models of disease transmission},
  author={van den Driessche, Pauline and Watmough, James},
  journal={Mathematical Biosciences},
  volume={180},
  number={1-2},
  pages={29--48},
  year={2002},
  publisher={Elsevier}
}

@article{GulbudakDengue,
title = {Infection severity across scales in multi-strain immuno-epidemiological Dengue model structured by host antibody level},
author = {Gulbudak, Hayriye and Browne, Cameron J.},
journal ={Journal of Mathematical Biology},
volume ={80},
number ={6},
year = {2020},
doi = {10.1007/s00285-020-01480-3}
}

@article{GulbudakRift,
author = {Tuncer, Necibe and Gulbudak, Hayriye and Cannataro, Vincent L. and Martcheva, Maia},
title ={Structural and Practical Identifiability Issues of Immuno-Epidemiological Vector–Host Models with Application to Rift Valley Fever},
journal ={Bulletin of Mathematical Biology},
year = {2016},
number ={9},
volume ={78},
doi = {10.1007/s11538-016-0200-2}
}

@article{gulbudak2022sensitivity,
  title={Sensitivity analysis in an immuno-epidemiological vector-host model},
  author={Gulbudak, Hayriye and Qu, Zhuolin and Milner, Fabio and Tuncer, Necibe},
  journal={Bulletin of mathematical biology},
  volume={84},
  number={2},
  pages={27},
  year={2022},
  publisher={Springer}
}

\end{document}